\documentclass[11pt,a4paper]{article}
\usepackage{amsmath, amsfonts, amssymb}
\usepackage{a4wide}
\usepackage[left=0.8in,bottom=1in,top=1in,right=0.8in]{geometry}
\usepackage{parskip}
\usepackage{enumitem}
\usepackage{xcolor}

\usepackage{hyperref}
\usepackage{booktabs}
\usepackage{caption}
\usepackage{siunitx}
\usepackage{tabularx}
\usepackage{comment}
\usepackage{graphicx}
\usepackage{adjustbox}
\usepackage{multirow,tensor}
\usepackage{amsthm}
\usepackage{bm, makecell}
\usepackage{lscape}
\usepackage{rotating}
\usepackage{floatrow}
\usepackage{bm}
\usepackage[noadjust]{cite}

\def\qed{\hfill {$\square$}\goodbreak \medskip}

\newtheorem{theorem}{Theorem}[section]
\newtheorem{lemma}[theorem]{Lemma}
\newtheorem{corollary}[theorem]{Corollary}

\newtheorem{definition}[theorem]{Definition}
\newtheorem{example}[theorem]{Example}
\newtheorem{remark}[theorem]{Remark}

\numberwithin{equation}{section}

\newcommand{\Tr}{\textnormal{Tr}}
\newcommand{\Supp}{\textnormal{Supp}}
\DeclareMathOperator{\wt}{wt}

\usepackage{tikz,xcolor,hyperref}
\usepackage{mathdots}

\definecolor{lime}{HTML}{A6CE39}
\DeclareRobustCommand{\orcidicon}{%
	\begin{tikzpicture}
		\draw[lime, fill=lime] (0,0) 
		circle [radius=0.16] 
		node[white] {{\fontfamily{qag}\selectfont \tiny ID}};
		\draw[white, fill=white] (-0.0625,0.095) 
		circle [radius=0.007];
	\end{tikzpicture}
	\hspace{-2mm}
}

\foreach \x in {A, ..., Z}{%
	\expandafter\xdef\csname orcid\x\endcsname{\noexpand\href{https://orcid.org/\csname orcidauthor\x\endcsname}{\noexpand\orcidicon}}
}

\begin{document}
	\date{}
		\title{Linear Codes over $\mathbb{F}_{q}+u\mathbb{F}_{q}$ associated with Simplicial Complexes, Their Gray Images, and Subfield Codes}

		\author{{\bf Ankit Yadav\footnote{email: {\tt ankityadav10102000@gmail.com}}\orcidA{},\; 
        \bf Akanksha Tiwari\footnote{email: {\tt akankshafzd8@gmail.com }}\orcidB{} and \bf
        Ritumoni Sarma\footnote{email: {\tt ritumoni407@gmail.com}}\orcidC{}} \\ $^{\ast \dagger \ddagger }$Department of Mathematics\\ Indian Institute of Technology Delhi\\Hauz Khas, New Delhi-110016, India\\
  }
  
\maketitle
\begin{abstract}
In recent years, simplicial complexes have gained considerable attention as a useful tool for constructing distance-optimal codes over finite fields. In this article, we construct four infinite families of linear codes over the ring $\mathcal{R}=\mathbb{F}_{q}+u\mathbb{F}_{q}$ with $u^2=0$ using simplicial complexes with one or two maximal elements, and completely determine their Lee weight distributions via exponential-sum techniques. By employing a Gray map on $\mathcal{R}$, we obtain infinite families of distance-optimal codes over $\mathbb{F}_{q}$, including a near-Griesmer family, and establish sufficient conditions for their minimality. Furthermore, we investigate the corresponding subfield codes and derive sufficient conditions for their distance-optimality and minimality, yielding infinite families of Griesmer and near-Griesmer codes.

\medskip

\noindent \textit{Keywords:} Simplicial complex, optimal code, minimal code, subfield code, Griesmer bound
			
\medskip
			
\noindent \textit{2020 Mathematics Subject Classification:} 94B05, 94B25, 94B60, 11T71

\end{abstract}

\section{Introduction}\label{Sec1}
Let $\mathbb{F}_{q}$ denote the finite field with $q$ elements, where $q$ is a prime power. A linear code with parameters $[n,k,d]$ is said to be distance-optimal if there exists no $[n, k, D]$-linear code over the same field with $D>d$. For fixed length and dimension, distance-optimal codes achieve the largest possible error-detection and error-correction capabilities. Thus, the construction of distance-optimal codes is a fundamental objective in coding theory. Various approaches have been developed for constructing such codes, one of which involves studying codes over finite rings and investigating their Gray images. The study of codes over finite rings has attracted considerable attention following the seminal article by Hammons et al. \cite{hammons1994Z4}, which demonstrated that several well-known binary non-linear codes can be realized as the Gray images of linear codes over $\mathbb{Z}_{4}$. Since then, researchers have investigated various families of codes over finite rings, including local rings, chain rings, and non-chain rings; see \cite{Bag:2022anr},\cite{abualrub2007cyclic},\cite{DINH2010940},\cite{atiwari} and the references therein. The study of the Lee weight distribution of codes over finite rings is of considerable interest. The weight distribution of a linear code plays a crucial role in determining its error-correcting capability and calculating the probabilities of error detection and correction \cite{klove_codes, peterson_cyclic}.  Furthermore, minimal linear codes have been extensively studied because of their applications in secret sharing schemes \cite{shamir1979secret, yuan2006secret} and secure two-party computation \cite{chabanne2014twoparty}. 

In \cite{ding2007cyclotomic}, the authors introduced a generic construction of linear codes over finite fields using trace functions. Subsequently, this approach has been widely employed to construct linear codes with good parameters; for instance, see \cite{ding2008two,ding2014binary,ding2015linear,heng2020optimal} and the references therein. This construction has also been extended to various finite rings. In \cite{shi2016optimal}, the authors studied trace codes over $\mathbb{F}_2+u\mathbb{F}_2$ with $u^2=0$. Trace codes over $\mathbb{F}_p+u\mathbb{F}_p$ with $u^2=0$ and $u^2=u$ were investigated in \cite{shi2017two} and \cite{shi2017twonew}, respectively. More recently, in \cite{liu2019two}, the authors studied the trace codes over $\mathbb{F}_q+u\mathbb{F}_q$ with $u^2=0$ and determined their Lee weight distributions. 

On the other hand, simplicial complexes have been shown to be an effective tool for constructing distance-optimal codes, as demonstrated by the following works in the literature. Chang and Hyun \cite{Chang_Hyun2018simplicial} were the first to utilize simplicial complexes in the construction of binary linear codes and obtained infinite families of optimal few-weight binary linear codes. Subsequently, in \cite{hyun2020infinite}, the authors employed the defining-set technique proposed by Ding and Niederreiter \cite{ding2007cyclotomic} to construct infinite families of optimal few-weight binary linear codes from simplicial complexes with one or two maximal elements. Since then, many researchers have constructed the distance-optimal codes over the finite field from simplicial complexes, see \cite{yadav2025optimal, vidya2024nonunital, Mondal2024mixed_alphabet, hu2024new, wu2024survey, chen2025griesmer, Hu2026several, Wu2026infinite, sharma2026infinite, jose2025linear}. Wu et al. \cite{wu2020optimal} constructed the linear codes over the ring $\mathbb{F}_{2}+u\mathbb{F}_{2}$ with $u^2=0$ from simplicial complexes with one maximal element and obtained two classes of binary optimal codes using a Gray map. Furthermore, Wu et al. \cite{Wu2020few_weight} constructed linear codes over $\mathbb{F}_{p}+u\mathbb{F}_{p}$ with $u^2=0$ simplicial complexes with one maximal element and obtained optimal linear codes over $\mathbb{F}_{p}$. Following these developments, simplicial complexes were employed to construct linear codes over $\mathbb{F}_{2}+u\mathbb{F}_{2}+u^2\mathbb{F}_{2}$ \cite{li2020new}, $\mathbb{F}_{2}+u\mathbb{F}_{2}+v\mathbb{F}_{2}+uv\mathbb{F}_{2}$ with $u^2=v^2=0$ and $uv=vu$ \cite{shi2022few-weight}, and $\mathbb{F}_{p}+u\mathbb{F}_{p}$ with $u^2=u$ \cite{shi2021two}, yielding optimal linear codes over finite fields through Gray maps. Most recently, Chen et al. \cite{chen2025optimal} constructed four families of linear codes over $\mathbb{F}_{q}+u\mathbb{F}_{q}$ with $u^2=0$ by using the simplicial complexes with one maximal element. They have determined their Lee weight distributions and obtained several families of distance-optimal codes over $\mathbb{F}_{q}$. 

The notion of subfield codes was first considered in \cite{canteaut2000weight,carlet1998codes}, although the term “subfield codes” was not used explicitly. The authors demonstrated that the subfield codes of several linear codes are distance-optimal. Subfield codes were formally introduced in \cite[p. 5117]{Magma-Handbook}. In \cite{Ding2019subfield}, the authors developed a general theory of subfield codes and studied subfield codes of two families of ovoid codes. For further studies on subfield codes of codes over various fields, one may refer to \cite{wu2022quaternary,sagar2023octanary, liu2024linear}. Recently, Bhagat et al. \cite{anuj2025subfield} investigated the subfield codes of linear codes constructed using simplicial complexes over the ring $\mathbb{F}_{2}+u\mathbb{F}_{2}+u^2\mathbb{F}_{2}$ with $u^3=u$. Furthermore, the authors in \cite{bhagat2025binary} studied subfield codes of codes over $\mathbb{F}_{2}[u]/\langle u^s\rangle$ for $s\geq 2$, while the authors in \cite{yadav2025optimal} considered subfield codes of codes over the ring $\mathbb{F}_{2}+u\mathbb{F}_{2}+v\mathbb{F}_{2}+uv\mathbb{F}_{2}$, where $u^2=v^2=0$ and $uv=vu$.

Motivated by the above developments, we in this article study the construction of linear codes over the ring $\mathcal{R}=\mathbb{F}_{q}+u\mathbb{F}_{q}$, where $u^2=0$, using simplicial complexes having one or two maximal elements, and determine their Lee-weight distributions. By employing a Gray map on $\mathcal{R}$, we obtain infinite families of few-weight distance-optimal codes over $\mathbb{F}_{q}$. We further establish sufficient conditions under which these Gray images are minimal. In addition, we investigate their subfield codes and derive conditions for their distance-optimality and minimality. To the best of our knowledge, this is the first study of subfield codes arising from codes over an $\mathbb{F}_{q}$-algebra for arbitrary $q$.

The remainder of this article is organized as follows. Section \ref{sec2} recalls the necessary preliminary concepts and results. In Section \ref{sec3}, we construct four infinite families of linear codes over $\mathcal{R}$ and determine their Lee-weight distributions. Section \ref{sec4} investigates their Gray images and establishes sufficient conditions for distance-optimality and minimality. In Section \ref{sec5}, we study the corresponding subfield codes and derive sufficient conditions for their distance-optimality and minimality. Finally, Section \ref{sec6} presents the concluding remarks.

\section{Preliminaries} \label{sec2}
Throughout this article, we use the following notation.
\begin{table}[H]
\centering
\begin{tabular}{l|l}
\hline
Notation & Description \\
\hline
$m$ & a positive integer \\
$[m]$ & the set $\{1, 2, \dots, m\}$\\
$\#S$ &  the cardinality of $S$ \\
$\mathcal{P}(S)$ & the power set of $S$\\
$\mathbb{F}_{q}$ & the finite field of size $q$ \\
$\mathcal{R}$ & the $\mathbb{F}_{q}$-algebra $\mathbb{F}_{q}+u\mathbb{F}_{q}$ with $u^2=0$ \\
$\Supp(\bm{v})$ & $\{j\in [m]: v_j\ne 0\}$, for $\bm{v}\in\mathbb{F}_{q}^{m}$\\
$\wt_{H}(\bm{v})$ & $\#\Supp (\bm{v})$\\
$\zeta_{p}$ & a primitive complex $p$-th root of unity\\
\hline
\end{tabular}
\end{table}
A \textit{ linear code} $\mathcal{C}$ with parameters $[n,k,d]$ of length $n$ over $\mathbb{F}_{q}$ is a $k$-dimensional subspace of $\mathbb{F}_{q}^{n}$, 
where $d$ is the minimum (Hamming) distance of $\mathcal{C}$. Let $A_{j}$ be the number of codewords of $\mathcal{C}$ having weight $j$, then the string $(A_{0},A_{1},\ldots,A_{n})$ is said to be the \textit{Hamming weight distribution} of $\mathcal{C}$ and the homogeneous polynomial 
\[
W_{\mathcal{C}}(x,y) = \sum\limits_{i=0}^{n} A_{i} x^{n-i}y^{i}
\]
is said to be the \textit{Hamming Weight Enumerator} of $\mathcal{C}$. The code $\mathcal{C}$  is called a $t$-\textit{weight} linear code if exactly $t$ entries among $A_{1},A_{2},\ldots,A_{n}$ are non-zero.
Let 
\[
g_q(k,d):=\sum_{i=0}^{k-1}\left\lceil\frac{d}{q^i}\right\rceil.
\]

The Griesmer bound for an $[n,k,d]$-linear code over $\mathbb{F}_{q}$ is $n\geq g_{q}(k,d)$. An $[n,k,d]$-linear code is called a \textit{Griesmer code} if $n=g_q(k,d)$, and a \textit{near-Griesmer code} if $n-1=g_q(k,d)$.
It is immediate that every Griesmer code is distance-optimal.
\subsection{Lee Weight and Gray Map}

 A \textit{linear code} of length $n$ over $\mathcal{R}$ is an $\mathcal{R}$-submodule of $\mathcal{R}^n$. We recall the Gray map $\phi:\mathcal{R}\to \mathbb{F}_{q}^{2}$ defined by
\[
a+ub \mapsto (b,a+b).
\]
This map extends componentwise to a map $\Phi:\mathcal{R}^{n}\to \mathbb{F}_{q}^{2n}$, defined by
\[
\bm{a}+u\bm{b} \mapsto (\bm{b},\bm{a}+\bm{b}),
\]
where $\bm{b}, \bm{c}\in \mathbb{F}_{q}^{n}.$ The \textit{Lee weight} of a vector $\bm{v}=\bm{a}+u \bm{b} \in \mathcal{R}^{n}$ is defined as 
\[
\wt_{L}(\bm{v}):= \wt_{H}(\Phi(\bm{v})) = \wt_{H}(\bm{b})+\wt_{H}(\bm{a}+\bm{b}).
\]
The \textit{Lee distance} between two vectors $\bm{u}$ and $\bm{v}$ is defined as 
\[
d_{L}(\bm{u},\bm{v}) := \wt_{L}(\bm{u}-\bm{v}).
\]
The Gray map $\Phi$ is a distance-preserving map from $\mathcal{R}^{n}$, equipped with the Lee distance, to $\mathbb{F}_{q}^{2n}$, equipped with the Hamming distance. Consequently, the Lee weight distribution of a linear code $\mathcal{C}$ over $\mathcal{R}$ coincides with the Hamming weight distribution of its Gray image $\Phi(\mathcal{C})$.

\subsection{Simplicial Complex}
Let $\bm{v},\bm{w} \in \mathbb{F}_{q}^{m}$. Then, $\bm{w}$ \textit{covers} $\bm{v}$ (we write $\bm{v} \preceq \bm{w}$) if $\Supp(\bm{w}) \supseteq \Supp(\bm{v})$. A \textit{simplicial complex} $\Delta$ is a subset of $ \mathbb{F}_{q}^{m}$ such that if $\bm{v}\in \Delta $, then $ \bm{w}\in \Delta$ for all $\bm{w}\preceq \bm{v}$.

 An element $\bm{v} \in \Delta$ with entries $0$ or $1$ is called \textit{maximal} if there is no $\bm{w}\in \Delta\setminus\{\bm{v}\}$ such that $\Supp(\bm{v})\subsetneq \Supp(\bm{w})$. Let $\mathcal{F}=\{F_{1}, F_{2},\ldots, F_{l}\}$ denote the set of maximal elements of $\Delta$. For each $i$, let $A_{i}=\Supp(F_{i})$, and define $\mathcal{A}=\{A_{1},A_{2},\ldots,A_{l}\}$. We refer to $\mathcal{A}$  as the support of $\Delta$. Observe that $\Delta$ is uniquely generated by $\mathcal{A}$.
A simplicial complex generated by a single maximal element $A\subseteq [m]$ (denoted by $\Delta_{A}$) is  $$\Delta_{A} := \{\bm{v}\in\mathbb{F}_{q}^{m}: \Supp(\bm{v}) \subseteq A\}.$$ Observe that $\Delta_A$ is a $|A|$-dimensional subspace of $\mathbb{F}_{q}^{m}$, and let $\Delta_{A}^{*} = \Delta_{A}\setminus \{\bm{0}\}$. Furthermore, a simplicial complex generated by two maximal elements $A$ and $B$, denoted by $\Delta_{A,B}$, is $\Delta_{A,B} := \Delta_{A} \cup \Delta_{B}$. Note that $ \Delta_{A} \cap \Delta_{B} = \Delta_{A\cap B}.$
For a simplicial complex $\Delta_{A}\subseteq \mathbb{F}_{q}^{m}$, it is easy to verify that $\Delta_{A}^{\perp}= \{\bm{v}\in\mathbb{F}_{q}^{m}: \Supp(\bm{v})\cap A = \emptyset\} = \Delta_{A^c}.$
\subsection{Auxiliary Lemmas}
Let $q=p^s$, where $p$ is a prime number and $s \geq 1$. Let $H$ be an $r$-dimensional $\mathbb{F}_{q}$-subspace of the vector space $\mathbb{F}_{q}^{m}$. The dual of $H$, denoted by $H^{\perp}$, is defined as 
\[
H^{\perp} = \{\bm{v}\in \mathbb{F}_{q}^{m}: \langle\bm{v},\bm{w}\rangle_{q}
=0 \text{ for all } \bm{w} \in H\},
\]
where $\langle\cdot,\cdot\rangle_{q}$ is the Euclidean inner product on $\mathbb{F}_{q}^{m}.$ Note that $H^{\perp}$ is $(m-r)$-dimensional $\mathbb{F}_{q}$-subspace of $\mathbb{F}_{q}^{m}$. 

Identify the field $\mathbb{F}_{q}$ with the vector space $\mathbb{F}_{p}^{s} $, and let $ \langle \cdot,\cdot \rangle_{p}$ denote the Euclidean inner product on  $\mathbb{F}_{p}^{s} $. Then, by \cite{Wu2026infinite}, we have the following lemma.
\begin{lemma}\label{lem:counting}
    Let $H$ be an $r$-dimensional  $\mathbb{F}_{q} $-subspace of  $\mathbb{F}_{q}^{m} $. Then for any $v \in  \mathbb{F}_{p}^{s}\setminus\{\bm{0}\} $, we have
    \[
    \sum\limits_{\bm{x} \in H} \zeta_{p}^{\langle v, \langle\bm{y},\bm{x}\rangle_{q} \rangle_{p}} = \begin{cases}
        q^r,\; \text{ if }\; \bm{y}\in H^{\perp}; \\
        0, \; \text{otherwise}.
    \end{cases} 
    \]
    
\end{lemma}

    Let $\mathcal{C}\subseteq \mathbb{F}_{q}^{m}$ be a linear code and $0\neq \bm{v} \in \mathcal{C}$. Then $\bm{v}$ is said to be \textit{minimal} if $\bm{w}\preceq \bm{v}$ implies $\bm{w} = \alpha \bm{v}$ for some $\alpha \in \mathbb{F}_{q}$. If every non-zero codeword of $\mathcal{C}$ is minimal, then $\mathcal{C}$ is called a \textit{minimal code}.
The following lemma, due to Ashikhmin and Barg \cite{Ashikhmin1998}, provides a sufficient condition for a linear code to be minimal.
\begin{lemma}\label{minimal_lemma}
    Let $\mathcal{C}$ be a linear code over $\mathbb{F}_{q}$. Then $\mathcal{C}$ is minimal if 
    \[
    \frac{\wt_{\min}}{\wt_{\max}}>\frac{q-1}{q},
    \]
     where $\wt_{\min}$ and $\wt_{\max}$ represent the minimum and maximum non-zero Hamming weights of $\mathcal{C}$, respectively.
\end{lemma}
The following result, due to \cite{hu2022subfield}, establishes the distance-optimality of a near-Griesmer code.
\begin{lemma}\label{lem:near-griesmer}
    Let $\mathcal{C}$ be an $[n,k,d]$ near-Griesmer code over $\mathbb{F}_{q}$ with $k>1$. Then $\mathcal{C}$ is distance-optimal if $q$ divides $d$.
\end{lemma}

The following lemma is useful in proving the optimality of the codes constructed in this article.
\begin{lemma}\label{lem: optimality}
  Let $t>t_{1}\geq t_{2} > 0$, and define
  \[
  S(t,t_{1},t_{2},q) := \sum\limits_{i=0}^{t-1} \left\lceil \frac{2(q-1)\left(q^{t-1}-q^{t_{1}-1}-q^{t_{2}-1}\right)+1}{q^{i}}\right\rceil.
  \]
  \begin{enumerate}
      \item If $t_{1} = t_{2}$, then
      \[
      S(t,t_{1},t_{2},q) = 
      \begin{cases}
          2q^{t}-2q^{t_1}-2q^{t_2}+t_2+1, & \text{if } q=2; \\
          2q^{t}-2q^{t_1}-2q^{t_2}+t_2, & \text{if } q=3 \text{ or } q=4; \\
          2q^{t}-2q^{t_1}-2q^{t_2}+t_2-1, & \text{if } q\geq 5.
      \end{cases}
      \]
      \item If $t_{1} > t_{2}$, then
      \[
      S(t,t_{1},t_{2},q) =
       \begin{cases}
          2q^{t}-2q^{t_1}-2q^{t_2}+t_2+1, & \text{if } q=2; \\
          2q^{t}-2q^{t_1}-2q^{t_2}+t_2, & \text{if } q\geq 3. 
      \end{cases}
      \]
  \end{enumerate}
\end{lemma}
\begin{proof}
    \begin{align*}
        S(t,t_{1},t_{2},q) &= \sum\limits_{i=0}^{t-1} \left\lceil \frac{2(q-1)\left(q^{t-1}-q^{t_{1}-1}-q^{t_{2}-1}\right)+1}{q^{i}}\right\rceil \\
        &= \sum\limits_{i=0}^{t_{2}-1} \frac{2(q-1)\left(q^{t-1}-q^{t_{1}-1}-q^{t_{2}-1}\right)}{q^{i}} + t_{2}
        +  \sum\limits_{i=t_2}^{t_{1}-1} \frac{2(q-1)\left(q^{t-1}-q^{t_{1}-1}\right)}{q^{i}} \\
        &\quad + \sum\limits_{i=t_2}^{t_{1}-1}\left\lceil -\frac{2(q-1)q^{t_2-1}-1}{q^i} \right\rceil + \sum\limits_{i=t_1}^{t-1} \frac{2(q-1)q^{t-1}}{q^{i}} + \sum\limits_{i=t_1}^{t-1}\left\lceil -\frac{2(q-1)(q^{t_1 - 1}+q^{t_2-1})-1}{q^i} \right\rceil \\
        &= 2q^{t}-2q^{t_1}-2q^{t_2}+t_{2}+2+ \sum\limits_{i=t_2}^{t_{1}-1}\left\lceil -\frac{2(q-1)q^{t_2-1}-1}{q^i} \right\rceil+\sum\limits_{i=t_1}^{t-1}\left\lceil -\frac{2(q-1)(q^{t_1 - 1}+q^{t_2-1})-1}{q^i} \right\rceil.
    \end{align*}
\begin{enumerate}
    \item Let $t_{1} = t_{2}$. Then
    \[
    \sum\limits_{i=t_2}^{t_{1}-1}\left\lceil -\frac{2(q-1)q^{t_2-1}-1}{q^i} \right\rceil = 0,
    \]
and,
 \[
    \sum\limits_{i=t_1}^{t-1}\left\lceil -\frac{2(q-1)(q^{t_1-1}+q^{t_2-1})-1}{q^i} \right\rceil = 
    \begin{cases}
        -1, & \text{if } q=2; \\
        -2, & \text{if } q=3 \text{ or } q=4; \\
        -3, & \text{if } q\geq 5. \\
    \end{cases}
    \]
    Therefore,
      \[
      S(t,t_{1},t_{2},q) = 
      \begin{cases}
          2q^{t}-2q^{t_1}-2q^{t_2}+t_2+1, & \text{if } q=2; \\
          2q^{t}-2q^{t_1}-2q^{t_2}+t_2, & \text{if } q=3 \text{ or } q=4; \\
          2q^{t}-2q^{t_1}-2q^{t_2}+t_2-1, & \text{if } q\geq 5.
      \end{cases}
      \]
      \item Let $t_1 > t_2$. Then
       \[
    \sum\limits_{i=t_2}^{t_{1}-1}\left\lceil -\frac{2(q-1)q^{t_2-1}-1}{q^i} \right\rceil = 
    \begin{cases}
        0, & \text{if } q=2;\\
        -1, & \text{if } q\geq 3.
    \end{cases}
    \]
and,
 \[
    \sum\limits_{i=t_1}^{t-1}\left\lceil -\frac{2(q-1)(q^{t_1-1}+q^{t_2-1})-1}{q^i} \right\rceil = -1 \text{ for all } q\geq 2.
    \]
Therefore,
      \[
      S(t,t_{1},t_{2},q) = 
      \begin{cases}
          2q^{t}-2q^{t_1}-2q^{t_2}+t_2+1, & \text{if } q=2; \\
          2q^{t}-2q^{t_1}-2q^{t_2}+t_2, & \text{if } q\geq 3.
      \end{cases}
      \]
\end{enumerate} \qed
\end{proof}

\subsection{\texorpdfstring{$\mathcal{C}_{D}$}{CD}-Construction and Weight Formula}
Let $D = D_{1}+uD_{2}\subseteq \mathcal{R}^{m}$, where $D_{1},D_{2} \subseteq \mathbb{F}_{q}^{m}$. Then define
\begin{equation}\label{eqn:C_D code}
    \mathcal{C}_{D}:= \{c_{D}(\bm{v}) = (\langle\bm{v},\bm{d}\rangle)_{\bm{d}\in D}: \bm{v}\in \mathcal{R}^{m}\},
\end{equation}
where $\langle\cdot,\cdot\rangle$ denotes the Euclidean inner product on $\mathcal{R}^{m}.$ Then $\mathcal{C}_{D}$ is a linear code over $\mathcal{R}$ with length $n=|D|$ and $D$ is called the defining set of $\mathcal{C}_{D}.$

Let $\bm{v}= \bm{a} +u \bm{b} \in \mathcal{R}^{m}$ and $\bm{d} = \bm{d}_{1}+u\bm{d}_{2} \in D$. Then
\begin{align*}
     \wt_{L}(c_{D}(\bm{v})) &= \wt_{L}\left(\langle\bm{a}+ u \bm{b},\bm{d}_{1}+u\bm{d}_{2}\rangle\right)_{\bm{d}_{1}\in D_{1},\bm{d}_{2}\in D_{2}}  \\ 
     &= \wt_{L}(\langle\bm{a},\bm{d}_{1}\rangle_{q}+u(\langle\bm{b},\bm{d}_{1}\rangle_{q}+\langle\bm{a},\bm{d}_{2}\rangle_{q}))_{\bm{d}_{1}\in D_{1},\bm{d}_{2}\in D_{2}}  \\
      &= \wt_{H}(\langle\bm{b},\bm{d}_{1}\rangle_{q}+\langle\bm{a},\bm{d}_{2}\rangle_{q})_{\bm{d}_{1}\in D_{1},\bm{d}_{2}\in D_{2}} + \wt_{H}(\langle\bm{a}+\bm{b},\bm{d}_{1}\rangle_{q}+\langle\bm{a},\bm{d}_{2}\rangle_{q})_{\bm{d}_{1}\in D_{1},\bm{d}_{2}\in D_{2}}.
\end{align*}
By the orthogonal property of nontrivial additive characters \cite{niederreiter1997finitefield}, we have
\begin{align} \label{eqn:weight_formula}
   \wt_{L}(c_{D}(\bm{v})) &= n-\frac{1}{q}\sum\limits_{v\in \mathbb{F}_{p}^{s}}\sum\limits_{\bm{d}_{1}\in D_{1}}\sum\limits_{\bm{d}_{2}\in D_{2}}\zeta_{p}^{\langle v, \langle\bm{b},\bm{d}_{1}\rangle_{q}+\langle\bm{a},\bm{d}_{2}\rangle_{q}\rangle_{p}} \nonumber \\
      &\quad  + n-\frac{1}{q}\sum\limits_{v\in \mathbb{F}_{p}^{s}}\sum\limits_{\bm{d}_{1}\in D_{1}}\sum\limits_{\bm{d}_{2}\in D_{2}}\zeta_{p}^{\langle v, \langle\bm{a}+\bm{b},\bm{d}_{1}\rangle_{q}+\langle\bm{a},\bm{d}_{2}\rangle_{q}\rangle_{p}} \nonumber \\
      &= 2n-\frac{1}{q}\sum\limits_{v\in \mathbb{F}_{p}^{s}}\sum\limits_{\bm{d}_{2}\in D_{2}}\zeta_{p}^{\langle v, \langle\bm{a},\bm{d}_{2}\rangle_{q}\rangle_{p}} \sum\limits_{\bm{d}_{1}\in D_{1}}\left(\zeta_{p}^{\langle v, \langle\bm{b},\bm{d}_{1}\rangle_{q}\rangle_p}+\zeta_{p}^{\langle v,\langle\bm{a}+\bm{b},\bm{d}_{1}\rangle_{q}\rangle_{p}}\right) \nonumber \\
      &= 2n-\frac{2}{q}n-\frac{1}{q}\sum\limits_{v\in \mathbb{F}_{p}^{s}\setminus\{0\}}\sum\limits_{\bm{d}_{2}\in D_{2}}\zeta_{p}^{\langle v, \langle\bm{a},\bm{d}_{2}\rangle_{q}\rangle_{p}} \sum\limits_{\bm{d}_{1}\in D_{1}}\left(\zeta_{p}^{\langle v, \langle\bm{b},\bm{d}_{1}\rangle_{q}\rangle_p}+\zeta_{p}^{\langle v,\langle\bm{a}+\bm{b},\bm{d}_{1}\rangle_{q}\rangle_{p}}\right).
\end{align}

\section{Lee Weight Distributions} \label{sec3}
In this section, we investigate four classes of $\mathcal{C}_{D}$-codes over the ring $\mathcal{R}$, where the defining set $D$ is constructed from a simplicial complex with one or two maximal elements.
\begin{theorem}\label{thm:1}
    Let $m\geq 2$ be a positive integer and $A,B \subseteq [m]$. Define $D = \Delta_{A}^* + u \Delta_{B}^*.$ Then the code $\mathcal{C}_{D}$ defined in \eqref{eqn:C_D code} is an at most five-weight linear code of length $(q^{|A|}-1)(q^{|B|}-1)$ and cardinality $q^{|A|+|A\cup B|}$. Its Lee weight distribution is given by 
    \begin{table}[H]
\centering
\resizebox{\textwidth}{!}{
\begin{tabular}{l|l}
\hline
Weight & Frequency \\
\hline
$w_0=0$ & $f_0=1$ \\
\hline
$w_{1}= (q-1)(q^{|A|+|B|-1}-q^{|A|-1})$ & $f_1=2(q^{|A\cup B|-|B|}-1)$ \\
\hline
$w_{2}= 2(q-1)(q^{|A|+|B|-1}-q^{|A|-1})$ & $f_2=q^{|A\cup B|+|A|-|B|}-2q^{|A\cup B|-|B|}+1$ \\
\hline
$w_{3}= 2(q-1)(q^{|A|+|B|-1}-q^{|B|-1})$ & $f_3=q^{|A\cup B|-|A|}-1$ \\
\hline
$w_{4}=  (q-1)(2q^{|A|+|B|-1}-q^{|A|-1}-2q^{|B|-1})$ & $f_4=2(q^{|A\cup B|}-q^{|A\cup B|-|A|}-q^{|A\cup B|-|B|}+1)$ \\
\hline
$w_{5}=  2(q-1)(q^{|A|+|B|-1}-q^{|A|-1}-q^{|B|-1})$ & $f_5=q^{|A|+|A\cup B|}-q^{|A|+|A\cup B|-|B|}+2q^{|A\cup B|-|B|}-2q^{|A\cup B|}+q^{|A\cup B|-|A|}-1$ \\
\hline
\end{tabular}
}
\label{tab:3}
\end{table}
\end{theorem}
\begin{proof}
    Here $D_{1}=\Delta_{A}^*, D_{2} = \Delta_{B}^*$. Consequently, the length of the code $\mathcal{C}_{D}$ is $n = |D| = (q^{|A|}-1)(q^{|B|}-1).$ Now, 
    \[
        \sum\limits_{\bm{d}_{2}\in \Delta_{B}^*}\zeta_{p}^{\langle v, \langle\bm{a},\bm{d}_{2}\rangle_{q}\rangle_{p}}
        = \sum\limits_{\bm{d}_{2}\in \Delta_{B}}\zeta_{p}^{\langle v, \langle\bm{a},\bm{d}_{2}\rangle_{q}\rangle_{p}} -  1,
    \]
    and,
     \[ \sum\limits_{\bm{d}_{1}\in \Delta_{A}^*}\left(\zeta_{p}^{\langle v, \langle\bm{b},\bm{d}_{1}\rangle_{q}\rangle_p}+\zeta_{p}^{\langle v, \langle\bm{a}+\bm{b},\bm{d}_{1}\rangle_{q}\rangle_{p}}\right)=  \sum\limits_{\bm{d}_{1}\in \Delta_{A}}\left(\zeta_{p}^{\langle v, \langle\bm{b},\bm{d}_{1}\rangle_{q}\rangle_p}+\zeta_{p}^{\langle v, \langle\bm{a}+\bm{b},\bm{d}_{1}\rangle_{q}\rangle_{p}}\right)-2.
 \]
    Then by Lemma \ref{lem:counting}, for $v\in\mathbb{F}_{p}^{s}\setminus\{0\}$, we have
 \begin{equation}\label{eqn:sum_Delta_B*}
     \sum\limits_{\bm{d}_{2}\in \Delta_{B}^*}\zeta_{p}^{\langle v, \langle\bm{a},\bm{d}_{2}\rangle_{q}\rangle_{p}} = 
     \begin{cases}
         q^{|B|}-1, & \text{if } \bm{a} \in \Delta_B^{\perp}; \\
         -1,  & \text{if } \bm{a} \notin \Delta_B^{\perp}.
     \end{cases}
 \end{equation}
Similarly, by Lemma \ref{lem:counting}, for $v\in\mathbb{F}_{p}^{s}\setminus\{0\}$, we have
\begin{equation}\label{eqn:sum_Delta_A*}
    \sum\limits_{\bm{d}_{1}\in \Delta_{A}}\left(\zeta_{p}^{\langle v, \langle\bm{b},\bm{d}_{1}\rangle_{q}\rangle_p}+\zeta_{p}^{\langle v, \langle\bm{a}+\bm{b},\bm{d}_{1}\rangle_{q}\rangle_{p}}\right) = 
\begin{cases}
    2q^{|A|}-2, & \text{if } \bm{b}\in \Delta_{A}^{\perp},\bm{a}+\bm{b}\in\Delta_{A}^{\perp};\\
    q^{|A|}-2, & \text{if } \bm{b}\notin \Delta_{A}^{\perp},\bm{a}+\bm{b}\in\Delta_{A}^{\perp} \;\text{ or }\; \bm{b}\in \Delta_{A}^{\perp},\bm{a}+\bm{b}\notin\Delta_{A}^{\perp} ;\\
    -2, & \text{if } \bm{b}\notin \Delta_{A}^{\perp},\bm{a}+\bm{b}\notin\Delta_{A}^{\perp}.
\end{cases}
\end{equation}

There are 2 possible cases which are as follows:
\begin{enumerate}
    \item[(1)] $\bm{a}\in \Delta_B^{\perp}$.
    \begin{itemize}
        \item If $\bm{b}\in \Delta_{A}^{\perp},\bm{a}+\bm{b}\in\Delta_{A}^{\perp},$ then
         \[
    \wt_{L}(c_{D}(\bm{v})) = 0.
    \]
    In this case, 
    \begin{eqnarray*}
        \# \bm{a} &=& |\Delta_A^{\perp}\cap\Delta_B^{\perp}|
        = q^{m-|A\cup B|}\\
    \# \bm{b} &=& |\Delta_{A}^{\perp}| = q^{m-|A|}.
    \end{eqnarray*}
     \item If $\bm{b}\in \Delta_{A}^{\perp},\bm{a}+\bm{b}\notin\Delta_{A}^{\perp},$ then
    \[\wt_{L}(c_{D}(\bm{v})) = (q-1)(q^{|A|+|B|-1}-q^{|A|-1}).\]
     In this case, 
    \begin{eqnarray*}
        \# \bm{a} &=& |\Delta_B^{\perp}|-|\Delta_B^{\perp}\cap\Delta_A^{\perp}| =q^{m-|B|}-q^{m-|A\cup B|}\\
    \# \bm{b} &=& |\Delta_{A}^{\perp}| = q^{m-|A|}.
    \end{eqnarray*}
    
    \item  If $\bm{b}\notin \Delta_{A}^{\perp},\bm{a}+\bm{b}\in\Delta_{A}^{\perp},$ then
    \[
    \wt_{L}(c_{D}(\bm{v})) = (q-1)(q^{|A|+|B|-1}-q^{|A|-1}).
    \]
     In this case, 
    \begin{eqnarray*}
        \# \bm{a} &=& q^{m-|B|}-q^{m-|A\cup B|}\\
    \# \bm{b} &=&  q^{m-|A|}.
    \end{eqnarray*}

    \item If $\bm{b}\notin \Delta_{A}^{\perp},\bm{a}+\bm{b}\notin\Delta_{A}^{\perp},$ then
    \[\wt_{L}(c_{D}(\bm{v})) = 2(q-1)(q^{|A|+|B|-1}-q^{|A|-1}).\]
     There are two subcases.
     \begin{itemize}
         \item $\bm{a}\in \Delta_{A}^{\perp}$. In this case, 
         \begin{eqnarray*}
        \# \bm{a} 
        &=& q^{m-|A\cup B|}\\
    \# \bm{b} &=& q^m - q^{m-|A|}.
    \end{eqnarray*}
    \item $\bm{a}\notin \Delta_{A}^{\perp}$. In this case, $\Delta_{A}^{\perp}$ and $\bm{a}+\Delta_{A}^{\perp}$ are two distinct cosets and hence
     \begin{eqnarray*}
        \# \bm{a} 
        &=&q^{m-|B|}-q^{m-|A\cup B|}\\
    \# \bm{b} &=& q^m-2q^{m-|A|}.
    \end{eqnarray*}
     \end{itemize}
   
    \end{itemize}

    \item[(2)]   $\bm{a}\notin \Delta_B^{\perp}$.
\begin{itemize}
        \item If $\bm{b}\in \Delta_{A}^{\perp},\bm{a}+\bm{b}\in\Delta_{A}^{\perp},$ then
         \[
    \wt_{L}(c_{D}(\bm{v})) = 2(q-1)(q^{|A|+|B|-1}-q^{|B|-1}).
    \]
    In this case, 
    \begin{eqnarray*}
        \# \bm{a} &=& |\Delta_A^{\perp}|-|\Delta_A^{\perp}\cap\Delta_B^{\perp}| = q^{m-|A|}-q^{m-|A\cup B|}\\
    \# \bm{b} &=& |\Delta_{A}^{\perp}| = q^{m-|A|}.
    \end{eqnarray*}

     \item If $\bm{b}\in \Delta_{A}^{\perp},\bm{a}+\bm{b}\notin\Delta_{A}^{\perp},$ then
    \[\wt_{L}(c_{D}(\bm{v})) = (q-1)(2q^{|A|+|B|-1}-q^{|A|-1}-2q^{|B|-1}).\]
     In this case, 
    \begin{eqnarray*}
        \# \bm{a} &=& q^m - q^{m-|A|}-q^{m-|B|}+q^{m-|A\cup B|}\\
    \# \bm{b} &=&  q^{m-|A|}.
    \end{eqnarray*}
    
    \item  If $\bm{b}\notin \Delta_{A}^{\perp},\bm{a}+\bm{b}\in\Delta_{A}^{\perp},$ then
    \[
    \wt_{L}(c_{D}(\bm{v})) = (q-1)(2q^{|A|+|B|-1}-q^{|A|-1}-2q^{|B|-1}). 
    \]
    In this case, 
    \begin{eqnarray*}
        \# \bm{a}
        &=&q^m-q^{m-|A|}-q^{m-|B|}+q^{m-|A\cup B|}\\
    \# \bm{b} &=& q^{m-|A|}.
    \end{eqnarray*}

    \item If $\bm{b}\notin \Delta_{A}^{\perp},\bm{a}+\bm{b}\notin\Delta_{A}^{\perp},$ then
    \[\wt_{L}(c_{D}(\bm{v})) = 2(q-1)(q^{|A|+|B|-1}-q^{|A|-1}-q^{|B|-1}).\]
   
 \end{itemize}
\end{enumerate}
To determine the frequency corresponding to the weight
\[
w_5 = 2(q-1)(q^{|A|+|B|-1}-q^{|A|-1}-q^{|B|-1}),
\]
we subtract the sum of the frequencies of all other weights from the total number of codewords $q^{2m}$. Consequently, the frequency is 
\[
q^{2m}-q^{2m-|B|}-2q^{2m-|A|}+q^{2m-2|A|}+2q^{2m-|A|-|B|}-q^{2m-|A|-|A\cup B|}.
\]
Since the frequency corresponding to Lee weight $0$ is $q^{2m-|A|-|A\cup B|}$, all the above frequencies must be divided by $q^{2m-|A|-|A\cup B|}$ to obtain the Lee weight distribution. \qed
\end{proof}
\begin{corollary}\label{cor:1}
     If $A\subseteq B$ in Theorem \ref{thm:1}, then the code $\mathcal{C}_{D}$ in \eqref{eqn:C_D code} is a four-weight linear code of length $(q^{|A|}-1)(q^{|B|}-1)$ and cardinality $q^{|A|+|B|}$. Its Lee weight distribution is given by
       \begin{table}[H]
\centering
\begin{tabular}{l|l}
\hline
Weight & Frequency \\
\hline
$0$ & $1$ \\
\hline
$2(q-1)(q^{|A|+|B|-1}-q^{|A|-1})$ & $q^{|A|}-1$ \\
\hline
$2(q-1)(q^{|A|+|B|-1}-q^{|B|-1})$ & $q^{|B|-|A|}-1$ \\
\hline
$(q-1)(2q^{|A|+|B|-1}-q^{|A|-1}-2q^{|B|-1})$ & $2(q^{|B|}-q^{|B|-|A|})$ \\
\hline
$2(q-1)(q^{|A|+|B|-1}-q^{|A|-1}-q^{|B|-1})$ & $q^{|A|+|B|}-q^{|A|}-2q^{|B|}+q^{|B|-|A|}+1$ \\
\hline
\end{tabular}
\label{tab:3c}
\end{table}
\end{corollary}

\begin{theorem}\label{thm:1a}
    Let $m\geq 2$ be a positive integer and $A,B \subseteq [m]$ such that $|A|+|B|<2m$. Define $D = (\Delta_{A} + u \Delta_{B})^c.$ Then the code $\mathcal{C}_{D}$ defined in \eqref{eqn:C_D code} is an at most three-weight linear code of length $q^{2m}-q^{|A|+|B|}$ and cardinality $q^{2m}$. Its Lee weight distribution is given by 
    \begin{table}[H]
\centering

\begin{tabular}{l|l}
\hline
Weight & Frequency \\
\hline
$w_0=0$ & $f_0=1$ \\
\hline
$w_{1}= 2(q-1)q^{2m-1}$ & $f_{1}=q^{2m-|A|-|A\cup B|}-1$ \\
\hline
$w_{2}= (q-1)(2q^{2m-1}-q^{|A|+|B|-1})$ & $f_{2}=2(q^{2m-|A|-|B|}-q^{2m-|A|-|A\cup B|})$ \\
\hline
$w_{3}= 2(q-1)(q^{2m-1}-q^{|A|+|B|-1})$ & $f_{3}=q^{2m}-2q^{2m-|A|-|B|}+q^{2m-|A|-|A\cup B|}$ \\
\hline
\end{tabular}

\label{tab:3a}
\end{table}
\end{theorem}
\begin{proof}
    The length of the code $\mathcal{C}_{D}$ is $n=|D| = q^{2m}-q^{|A|+|B|}$. Observe that 
    \[
    D = \left(\Delta_{A}^{c}+u\mathbb{F}_{q}^{m}\right)\bigsqcup \left(\Delta_{A}+u\Delta_{B}^{c}\right).
    \]
    \begin{align*}
        \wt_{L}(c_{D}(\bm{v}))&= 2n-\frac{2}{q}n-\frac{1}{q}\sum\limits_{v\in \mathbb{F}_{p}^{s}\setminus\{0\}}\sum\limits_{\bm{d}_{2}\in \mathbb{F}_{q}^{m}}\zeta_{p}^{\langle v, \langle\bm{a},\bm{d}_{2}\rangle_{q}\rangle_{p}} \sum\limits_{\bm{d}_{1}\in \mathbb{F}_{q}^{m}}\left(\zeta_{p}^{\langle v, \langle\bm{b},\bm{d}_{1}\rangle_{q}\rangle_p}+\zeta_{p}^{\langle v,\langle\bm{a}+\bm{b},\bm{d}_{1}\rangle_{q}\rangle_{p}}\right) \\
        & \quad + \frac{1}{q}\sum\limits_{v\in \mathbb{F}_{p}^{s}\setminus\{0\}}\sum\limits_{\bm{d}_{2}\in \Delta_{B}}\zeta_{p}^{\langle v, \langle\bm{a},\bm{d}_{2}\rangle_{q}\rangle_{p}} \sum\limits_{\bm{d}_{1}\in \Delta_{A}}\left(\zeta_{p}^{\langle v, \langle\bm{b},\bm{d}_{1}\rangle_{q}\rangle_p}+\zeta_{p}^{\langle v,\langle\bm{a}+\bm{b},\bm{d}_{1}\rangle_{q}\rangle_{p}}\right).
    \end{align*}
    There are $2$ possible cases which are as follows:
    \begin{enumerate}
        \item $\bm{a}=\bm{0}.$
        \begin{itemize}
            \item If $\bm{b}=\bm{0}$, then 
            \[
            \wt_{L}(c_{D}(\bm{v})) = 0.
            \]
            In this case,
            \begin{eqnarray*}
        \# \bm{a} &=& 1\\
    \# \bm{b} &=& 1.
    \end{eqnarray*}
    \item If $\bm{b}\neq \bm{0}, \bm{b}\in \Delta_{A}^{\perp}$, then
     \[
            \wt_{L}(c_{D}(\bm{v})) = 2(q-1)q^{2m-1}.
            \]
            In this case,
            \begin{eqnarray*}
        \# \bm{a} &=& 1\\
    \# \bm{b} &=& q^{m-|A|}-1.
    \end{eqnarray*}
    \item If $\bm{b}\notin \Delta_{A}^{\perp}$, then
     \[
            \wt_{L}(c_{D}(\bm{v})) = 2(q-1)(q^{2m-1}-q^{|A|+|B|-1}).
            \]
        \end{itemize}
        \item $\bm{a}\neq \bm{0}$ and $\bm{a}\in \Delta_{B}^{\perp}$.
        \begin{itemize}
            \item If $\bm{b}\in \Delta_{A}^{\perp}, \bm{a}+\bm{b}\in \Delta_{A}^{\perp}$, then
             \[
            \wt_{L}(c_{D}(\bm{v})) = 2(q-1)q^{2m-1}.
            \]
            In this case,
            \begin{eqnarray*}
        \# \bm{a} &=& q^{m-|A\cup B|}-1\\
    \# \bm{b} &=& q^{m-|A|}.
    \end{eqnarray*}
    \item If $\bm{b}\in \Delta_{A}^{\perp}, \bm{a}+\bm{b}\notin \Delta_{A}^{\perp}$, then
             \[
            \wt_{L}(c_{D}(\bm{v})) = (q-1)(2q^{2m-1}-q^{|A|+|B|-1}).
            \]
            In this case,
            \begin{eqnarray*}
        \# \bm{a} &=& q^{m-|B|}-q^{m-|A\cup B|}\\
    \# \bm{b} &=& q^{m-|A|}.
    \end{eqnarray*}
    \item If $\bm{b}\notin \Delta_{A}^{\perp}, \bm{a}+\bm{b}\in \Delta_{A}^{\perp}$, then
             \[
            \wt_{L}(c_{D}(\bm{v})) = (q-1)(2q^{2m-1}-q^{|A|+|B|-1}).
            \]
            In this case,
            \begin{eqnarray*}
        \# \bm{a} &=& q^{m-|B|}-q^{m-|A\cup B|}\\
    \# \bm{b} &=& q^{m-|A|}.
    \end{eqnarray*}
    \item  If $\bm{b}\notin \Delta_{A}^{\perp}, \bm{a}+\bm{b}\notin \Delta_{A}^{\perp}$, then
             \[
            \wt_{L}(c_{D}(\bm{v})) = 2(q-1)(q^{2m-1}-q^{|A|+|B|-1}).
            \]
        \end{itemize}
        \item $\bm{a}\notin \Delta_{B}^{\perp}$. Then 
         \[
            \wt_{L}(c_{D}(\bm{v})) = 2(q-1)(q^{2m-1}-q^{|A|+|B|-1}).
            \]
    \end{enumerate}
    To determine the frequency corresponding to the weight
\[
w_3 = 2(q-1)(q^{2m-1}-q^{|A|+|B|-1}),
\]
we subtract the sum of the frequencies of all other weights from the total number of codewords $q^{2m}$. Consequently, the frequency is 
\[
q^{2m}-2q^{2m-|A|-|B|}+q^{2m-|A|-|A\cup B|}.
\]
\qed
\end{proof}
\begin{corollary}\label{cor:2}
     If $A\subseteq B$ in Theorem \ref{thm:1a}, then the code $\mathcal{C}_{D}$ in \eqref{eqn:C_D code} is a two-weight linear code of length $q^{2m}-q^{|A|+|B|}$ and cardinality $q^{2m}$. Its Lee weight distribution is given by
     \begin{table}[H]
\centering

\begin{tabular}{l|l}
\hline
Weight & Frequency \\
\hline
$0$ & $1$ \\
\hline
$ 2(q-1)q^{2m-1}$ & $q^{2m-|A|-|B|}-1$ \\
\hline
$ 2(q-1)(q^{2m-1}-q^{|A|+|B|-1})$ & $q^{2m}-q^{2m-|A|-|B|}$ \\
\hline
\end{tabular}

\label{tab:3b}
\end{table}
\end{corollary}

\begin{theorem}\label{thm:2}
    Let $m\geq 2$ be a positive integer and let $A,B,C \subseteq [m]$ such that       $|C|\leq |B|$. Define $D = \Delta_{A}+ u (\Delta_{B,C}\setminus \Delta_{B\cap C})^{c}.$ Then the code $\mathcal{C}_{D}$ defined in \eqref{eqn:C_D code} is an at most nine-weight linear code of length $q^{|A|} (q^{m}-q^{|B|}-q^{|C|}+2q^{|B \cap C|})$ and cardinality $q^{m+|A|}$.
    Its Lee weight distribution is given by 
    \begin{table}[H]
\centering
\resizebox{\textwidth}{!}{
\begin{tabular}{l|l}
\hline
Weight & Frequency \\
\hline
$w_{0}=0$ & $f_{0}=1$ \\
\hline
$w_{1}= 2(q-1)q^{m+|A|-1}$ & $f_{1}=q^{m-|A\cup B\cup C|}-1$ \\
\hline
$w_{2}= (q-1)q^{|A|-1}(2q^{m}-q^{|B|}-q^{|C|}+2q^{|B\cap C|})$ & $f_{2}=2(q^{m-|B\cup C|}-q^{m-|A\cup B\cup C|})$ \\
\hline
$w_{3}= 2(q-1)q^{|A|-1}(q^{m}-q^{|C|})$ & $f_{3}=q^{m-|A\cup B|}-q^{m-|A\cup B\cup C|}$ \\
\hline
$w_{4}= (q-1)q^{|A|-1}(2q^{m}-q^{|B|}-2q^{|C|}+2q^{|B\cap C|})$ & $f_{4}=2(q^{m-|B|}-q^{m-|A\cup B|}-q^{m-|B\cup C|}+q^{m-|A\cup B\cup C|})$ \\
\hline
$w_{5}= 2(q-1)q^{|A|-1}(q^{m}-q^{|B|})$ & $f_{5}=q^{m-|A\cup C|}-q^{m-|A\cup B\cup C|}$ \\
\hline
$w_{6}= (q-1)q^{|A|-1}(2q^{m}-2q^{|B|}-q^{|C|}+2q^{|B\cap C|})$ & $f_{6}=2(q^{m-|C|}-q^{m-|A\cup C|}-q^{m-|B\cup C|}+q^{m-|A\cup B\cup C|})$ \\
\hline
$w_{7}= 2(q-1)q^{|A|-1}(q^{m}-q^{|B|}-q^{|C|})$ & $f_{7}=q^{m-|A\cup(B\cap C)|}-q^{m-|A\cup B|}-q^{m-|A\cup C|}+q^{m-|A\cup B\cup C|}$ \\
\hline
$w_{8}= 2(q-1)q^{|A|-1}(q^{m}-q^{|B|}-q^{|C|}+q^{|B\cap C|})$ &
$
\begin{aligned}
   f_{8}= 2(& q^{m-|B\cap C|}-q^{m-|A\cup(B\cap C)|}-q^{m-|B|}-q^{m-|C|}+q^{m-|A\cup B|} \\[-1ex]
    & +q^{m-|A\cup C|}+q^{m-|B\cup C|}-q^{m-|A\cup B\cup C|})
\end{aligned}$\\
\hline
$w_{9}= 2(q-1)q^{|A|-1}(q^{m}-q^{|B|}-q^{|C|}+2q^{|B\cap C|})$ & $f_{9}=q^{m+|A|}-2q^{m-|B\cap C|}+q^{m-|A\cup (B\cap C)|}$ \\
\hline
\end{tabular}
}
\label{tab:1}
\end{table}
\end{theorem}
\begin{proof}
    Here $D_{1}=\Delta_{A}, D_{2} = (\Delta_{B,C}\setminus \Delta_{B\cap C})^{c}$. Consequently, the length of the code $\mathcal{C}_{D}$ is $n = |D| = q^{|A|}(q^m-q^{|B|}-q^{|C|}+2q^{|B\cap C|}).$ Now, 
    \[
        \sum\limits_{\bm{d}_{2}\in D_{2}}\zeta_{p}^{\langle v, \langle\bm{a},\bm{d}_{2}\rangle_{q}\rangle_{p}}
        = \sum\limits_{\bm{d}_{2}\in \mathbb{F}_{q}^{m}}\zeta_{p}^{\langle v, \langle\bm{a},\bm{d}_{2}\rangle_{q}\rangle_{p}} -  \sum\limits_{\bm{d}_{2}\in \Delta_{B}}\zeta_{p}^{\langle v, \langle\bm{a},\bm{d}_{2}\rangle_{q}\rangle_{p}} -  \sum\limits_{\bm{d}_{2}\in \Delta_{C}}\zeta_{p}^{\langle v, \langle\bm{a},\bm{d}_{2}\rangle_{q}\rangle_{p}} + 2 \sum\limits_{\bm{d}_{2}\in \Delta_{B\cap C}}\zeta_{p}^{\langle v, \langle\bm{a},\bm{d}_{2}\rangle_{q}\rangle_{p}}.
    \]
    Then by Lemma \ref{lem:counting}, for $v\in\mathbb{F}_{p}^{s}\setminus\{0\}$, we have
 \begin{equation}\label{eqn:sum_Delta_BC}
     \sum\limits_{\bm{d}_{2}\in (\Delta_{B,C}\setminus \Delta_{B\cap C})^{c}}\zeta_{p}^{\langle v, \langle\bm{a},\bm{d}_{2}\rangle_{q}\rangle_{p}} = 
     \begin{cases}
         q^{m}-q^{|B|}-q^{|C|}+2q^{|B\cap C|}, & \text{if } \bm{a} = \bm{0}; \\
       -q^{|B|}-q^{|C|}+2q^{|B\cap C|},  & \text{if } \bm{a} \neq \bm{0}, \bm{a}\in\Delta_{B}^{\perp}, \bm{a}\in\Delta_{C}^{\perp}; \\
        -q^{|B|}+2q^{|B\cap C|},  & \text{if }  \bm{a}\in\Delta_{B}^{\perp}, \bm{a}\notin\Delta_{C}^{\perp}; \\
        -q^{|C|}+2q^{|B\cap C|},  & \text{if }  \bm{a}\notin\Delta_{B}^{\perp}, \bm{a}\in\Delta_{C}^{\perp}; \\
        2q^{|B\cap C|},  & \text{if }  \bm{a}\notin\Delta_{B}^{\perp}, \bm{a}\notin\Delta_{C}^{\perp}, \bm{a}\in\Delta_{B\cap C}^{\perp}; \\
        0,  & \text{if }   \bm{a}\in\Delta_{B\cap C}^{\perp}. 
     \end{cases}
 \end{equation}
Similarly, by Lemma \ref{lem:counting}, for $v\in\mathbb{F}_{p}^{s}\setminus\{0\}$, we have
\begin{equation}\label{eqn:sum_Delta_A}
    \sum\limits_{\bm{d}_{1}\in \Delta_{A}}\left(\zeta_{p}^{\langle v, \langle\bm{b},\bm{d}_{1}\rangle_{q}\rangle_p}+\zeta_{p}^{\langle v, \langle\bm{a}+\bm{b},\bm{d}_{1}\rangle_{q}\rangle_{p}}\right) = 
\begin{cases}
    2q^{|A|}, & \text{if } \bm{b}\in \Delta_{A}^{\perp},\bm{a}+\bm{b}\in\Delta_{A}^{\perp};\\
    q^{|A|}, & \text{if } \bm{b}\notin \Delta_{A}^{\perp},\bm{a}+\bm{b}\in\Delta_{A}^{\perp} \;\text{ or }\; \bm{b}\in \Delta_{A}^{\perp},\bm{a}+\bm{b}\notin\Delta_{A}^{\perp} ;\\
    0, & \text{if } \bm{b}\notin \Delta_{A}^{\perp},\bm{a}+\bm{b}\notin\Delta_{A}^{\perp}.
\end{cases}
\end{equation}

There are $6$ possible cases which are as follows:
\begin{enumerate}
    \item[(1)] $\bm{a}=\mathbf{0}$.
    \begin{itemize}
        \item If $\bm{b}\in \Delta_{A}^{\perp},$ then
         \[
    \wt_{L}(c_{D}(\bm{v})) = 0.
    \]
    In this case, 
    \begin{eqnarray*}
        \# \bm{a} &=& 1 \\
    \# \bm{b} &=& |\Delta_{A}^{\perp}| = q^{m-|A|}.
    \end{eqnarray*}
    
    \item  If $\bm{b}\notin \Delta_{A}^{\perp},$ then
    \[
    \wt_{L}(c_{D}(\bm{v})) = 2(q-1)q^{|A|-1}(q^{m}-q^{|B|}-q^{|C|}+2q^{|B \cap C|}).
    \]
    \end{itemize}

    \item[(2)]  $\bm{a}\neq \mathbf{0}, \bm{a}\in \Delta_{B}^{\perp}$ and $\bm{a}\in \Delta_{C}^{\perp}$.
\begin{itemize}
    \item If $\bm{b} \in \Delta_{A}^{\perp}$ and $\bm{a}+\bm{b} \in \Delta_{A}^{\perp}$, then  
    \[
    \wt_{L}(c_{D}(\bm{v})) = 2(q-1)q^{m+|A|-1}.
    \]
    In this case, 
    \begin{eqnarray*}
        \# \bm{a} &=& |\Delta_{A}^{\perp}\cap \Delta_{B}^{\perp}\cap \Delta_{C}^{\perp}| -1 = q^{m-|A\cup B \cup C|}-1 \\
    \# \bm{b} &=& |\Delta_{A}^{\perp}| = q^{m-|A|}.
    \end{eqnarray*}
    
    \item If $\bm{b} \in \Delta_{A}^{\perp}$ and $\bm{a}+\bm{b} \notin \Delta_{A}^{\perp}$, then 
    \[
    \wt_{L}(c_{D}(\bm{v})) = (q-1)q^{|A|-1}(2q^{m}-q^{|B|}-q^{|C|}+2q^{|B \cap C|}).
    \]
    In this case, 
    \begin{eqnarray*}
        \# \bm{a} 
       &=& |\Delta_{B}^{\perp} \cap \Delta_{C}^{\perp}| -|\Delta_{A}^{\perp}\cap \Delta_{B}^{\perp}\cap \Delta_{C}^{\perp}|
        = q^{m-|B \cup C|}-q^{m-|A\cup B \cup C|} \\
    \# \bm{b}
    &=& |\Delta_{A}^{\perp}| = q^{m-|A|}.
    \end{eqnarray*}

    \item If $\bm{b} \notin \Delta_{A}^{\perp}$ and $\bm{a}+\bm{b} \in \Delta_{A}^{\perp}$, then  
    \[
    \wt_{L}(c_{D}(\bm{v})) =  (q-1)q^{|A|-1}(2q^{m}-q^{|B|}-q^{|C|}+2q^{|B \cap C|}).
    \]
    In this case, 
    \begin{eqnarray*}
        \# \bm{a} 
        &=& q^{m-|B \cup C|}-q^{m-|A\cup B \cup C|} \\
    \# \bm{b}
    &=& q^{m-|A|}.
    \end{eqnarray*}

    \item If $\bm{b} \notin \Delta_{A}^{\perp}$ and $\bm{a}+\bm{b} \notin \Delta_{A}^{\perp}$, then 
    \[
    \wt_{L}(c_{D}(\bm{v})) = 2(q-1)q^{|A|-1}(q^{m}-q^{|B|}-q^{|C|}+2q^{|B \cap C|}).
    \]
    \end{itemize} 

    \item[(3)] $\bm{a}\in \Delta_{B}^{\perp}$ and $\bm{a}\notin \Delta_{C}^{\perp}$. 
\begin{itemize}
    \item If $\bm{b} \in \Delta_{A}^{\perp}$ and $\bm{a}+\bm{b} \in \Delta_{A}^{\perp}$, then 
    \[
    \wt_{L}(c_{D}(\bm{v})) = 2(q-1)q^{|A|-1}(q^{m}-q^{|C|}).
    \]
    In this case, 
    \begin{eqnarray*}
        \# \bm{a} 
        &=& |\Delta_{A}^{\perp}\cap \Delta_{B}^{\perp}|-|\Delta_{A}^{\perp}\cap \Delta_{B}^{\perp}\cap \Delta_{C}^{\perp}| = q^{m-|A\cup B|}-q^{m-|A\cup B \cup C|} \\
    \# \bm{b}
    &=& q^{m-|A|}.
    \end{eqnarray*}
    
    \item If $\bm{b} \in \Delta_{A}^{\perp}$ and $\bm{a}+\bm{b} \notin \Delta_{A}^{\perp}$, then
    \[
    \wt_{L}(c_{D}(\bm{v})) = (q-1)q^{|A|-1}(2q^{m}-q^{|B|}-2q^{|C|}+2q^{|B \cap C|}).
    \]
    In this case, 
    \begin{eqnarray*}
        \# \bm{a}
        &=& |\Delta_{B}^{\perp}|-|\Delta_{B}^{\perp} \cap \Delta_{A}^{\perp}|-|\Delta_{B}^{\perp} \cap \Delta_{C}^{\perp}| +|\Delta_{A}^{\perp}\cap \Delta_{B}^{\perp}\cap \Delta_{C}^{\perp}|
        \\ 
        &=& q^{m-|B|}-q^{m-|A\cup B|}-q^{m-|B \cup C|}+q^{m-|A\cup B \cup C|} \\
    \# \bm{b} &=&|\Delta_{A}^{\perp}| = q^{m-|A|}.
    \end{eqnarray*}

    \item If $\bm{b} \notin \Delta_{A}^{\perp}$ and $\bm{a}+\bm{b} \in \Delta_{A}^{\perp}$, then
    \[
    \wt_{L}(c_{D}(\bm{v})) = (q-1)q^{|A|-1}(2q^{m}-q^{|B|}-2q^{|C|}+2q^{|B \cap C|}).
    \]
    In this case, 
    \begin{eqnarray*}
        \# \bm{a}
        &=& q^{m-|B|}-q^{m-|A\cup B|}-q^{m-|B \cup C|}+q^{m-|A\cup B \cup C|} \\
    \# \bm{b} &=& q^{m-|A|}.
    \end{eqnarray*}

    \item If $\bm{b} \notin \Delta_{A}^{\perp}$ and $\bm{a}+\bm{b} \notin \Delta_{A}^{\perp}$, then 
    \[
    \wt_{L}(c_{D}(\bm{v})) = 2(q-1)q^{|A|-1}(q^{m}-q^{|B|}-q^{|C|}+2q^{|B \cap C|}).
    \]    
\end{itemize}

\item[(4)] $\bm{a}\notin \Delta_{B}^{\perp}$ and $\bm{a}\in \Delta_{C}^{\perp}$. 
\begin{itemize}
    \item If $\bm{b} \in \Delta_{A}^{\perp}$ and $\bm{a}+\bm{b} \in \Delta_{A}^{\perp}$, then 
    \[
    \wt_{L}(c_{D}(\bm{v})) = 2(q-1)q^{|A|-1}(q^{m}-q^{|B|}).
    \]
    In this case, 
    \begin{eqnarray*}
        \# \bm{a}
        &=&|\Delta_{A}^{\perp}\cap \Delta_{C}^{\perp}|-|\Delta_{A}^{\perp}\cap \Delta_{B}^{\perp}\cap \Delta_{C}^{\perp}| = q^{m-|A\cup C|}-q^{m-|A\cup B \cup C|} \\
    \# \bm{b} &=& |\Delta_{A}^{\perp}| = q^{m-|A|}.
    \end{eqnarray*}
    
    \item If $\bm{b} \in \Delta_{A}^{\perp}$ and $\bm{a}+\bm{b} \notin \Delta_{A}^{\perp}$, then 
    \[
    \wt_{L}(c_{D}(\bm{v})) = (q-1)q^{|A|-1}(2q^{m}-2q^{|B|}-q^{|C|}+2q^{|B \cap C|}).
    \]
    In this case, 
    \begin{eqnarray*}
        \# \bm{a} &=& |\Delta_{C}^{\perp}|-|\Delta_{C}^{\perp} \cap \Delta_{A}^{\perp}|-|\Delta_{C}^{\perp} \cap \Delta_{B}^{\perp}| +|\Delta_{A}^{\perp}\cap \Delta_{B}^{\perp}\cap \Delta_{C}^{\perp}|
        \\ 
        &=& q^{m-|C|}-q^{m-|A\cup C|}-q^{m-|B \cup C|}+q^{m-|A\cup B \cup C|} \\
    \# \bm{b} &=& |\Delta_{A}^{\perp}| = q^{m-|A|}.
    \end{eqnarray*}

    \item If $\bm{b} \notin \Delta_{A}^{\perp}$ and $\bm{a}+\bm{b} \in \Delta_{A}^{\perp}$, then
    \[
    \wt_{L}(c_{D}(\bm{v})) = (q-1)q^{|A|-1}(2q^{m}-2q^{|B|}-q^{|C|}+2q^{|B \cap C|}).
    \]
    In this case, 
    \begin{eqnarray*}
        \# \bm{a} &=& q^{m-|C|}-q^{m-|A\cup C|}-q^{m-|B \cup C|}+q^{m-|A\cup B \cup C|} \\
    \# \bm{b} &=& q^{m-|A|}.
    \end{eqnarray*}

    \item If $\bm{b} \notin \Delta_{A}^{\perp}$ and $\bm{a}+\bm{b} \notin \Delta_{A}^{\perp}$, then 
    \[
    \wt_{L}(c_{D}(\bm{v})) = 2(q-1)q^{|A|-1}(q^{m}-q^{|B|}-q^{|C|}+2q^{|B \cap C|}).
    \]    
\end{itemize}

\item[(5)] $\bm{a}\notin \Delta_{B}^{\perp},\bm{a}\notin \Delta_{C}^{\perp}$ and $\bm{a}\in\Delta_{B\cap C}^{\perp}$. 
\begin{itemize}
    \item If $\bm{b} \in \Delta_{A}^{\perp}$ and $\bm{a}+\bm{b} \in \Delta_{A}^{\perp}$, then 
    \[
    \wt_{L}(c_{D}(\bm{v})) = 2(q-1)q^{|A|-1}(q^{m}-q^{|B|}-q^{|C|}).
    \]
    In this case, 
    \begin{eqnarray*}
        \# \bm{a} &=& |\Delta_{A}^{\perp}\cap \Delta_{B\cap C}^{\perp}|-|\Delta_{A}^{\perp}\cap \Delta_{B}^{\perp}|-|\Delta_{A}^{\perp}\cap \Delta_{C}^{\perp}|+|\Delta_{A}^{\perp}\cap \Delta_{B}^{\perp}\cap\Delta_{C}^{\perp}|
        \\ 
        &=& q^{m-|A\cup (B\cap C)|}-q^{m-|A\cup B|}-q^{m-|A \cup C|}+q^{m-|A\cup B \cup C|} \\
    \# \bm{b} &=& |\Delta_{A}^{\perp}| = q^{m-|A|}.
    \end{eqnarray*}
    
    \item If $\bm{b} \in \Delta_{A}^{\perp}$ and $\bm{a}+\bm{b} \notin \Delta_{A}^{\perp}$, then 
    \[
    \wt_{L}(c_{D}(\bm{v})) = 2(q-1)q^{|A|-1}(q^{m}-q^{|B|}-q^{|C|}+q^{|B \cap C|}).
    \]
    In this case, 
    \[
\begin{aligned}
\#\bm{a}
&= |\Delta_{B\cap C}^{\perp}|-|\Delta_{B\cap C}^{\perp} \cap \Delta_{A}^{\perp}|-| \Delta_{B}^{\perp}|-|\Delta_{C}^{\perp}|+|\Delta_{A}^{\perp}\cap \Delta_{B}^{\perp}| \\
        &\quad +|\Delta_{A}^{\perp}\cap \Delta_{C}^{\perp}|+|\Delta_{B}^{\perp}\cap \Delta_{C}^{\perp}| -|\Delta_{A}^{\perp}\cap \Delta_{B}^{\perp}\cap \Delta_{C}^{\perp}| \\
        &= q^{m-|B\cap C|}-q^{m-|A\cup (B\cap C)|}-q^{m-|B|}-q^{m-|C|}+q^{m-|A\cup B|} \\
        &\quad +q^{m-|A \cup C|}+q^{m-|B\cup C|}-q^{m-|A\cup B \cup C|} \\
\#\bm{b}
&= |\Delta_{A}^{\perp}| = q^{m-|A|}.
\end{aligned}
\]
    \item If $\bm{b} \notin \Delta_{A}^{\perp}$ and $\bm{a}+\bm{b} \in \Delta_{A}^{\perp}$, then 
    \[
    \wt_{L}(c_{D}(\bm{v})) =  2(q-1)q^{|A|-1}(q^{m}-q^{|B|}-q^{|C|}+q^{|B \cap C|}).
    \]
    In this case, 
    \[
\begin{aligned}
\#\bm{a}
&= q^{m-|B\cap C|}
 - q^{m-|A\cup(B\cap C)|}
 - q^{m-|B|}
 - q^{m-|C|} + q^{m-|A\cup B|} \\
&\quad
 + q^{m-|A\cup C|}
 + q^{m-|B\cup C|}
 - q^{m-|A\cup B\cup C|} \\
\#\bm{b}
&= q^{m-|A|}.
\end{aligned}
\]
    
    \item If $\bm{b} \notin \Delta_{A}^{\perp}$ and $\bm{a}+\bm{b} \notin \Delta_{A}^{\perp}$, then 
    \[
    \wt_{L}(c_{D}(\bm{v})) = 2(q-1)q^{|A|-1}(q^{m}-q^{|B|}-q^{|C|}+2q^{|B \cap C|}).
    \]    
\end{itemize}

\item[(6)] $\bm{a}\notin \Delta_{B\cap C}^{\perp}$. Then 
\[
\wt_{L}(c_{D}(\bm{v})) = 2(q-1)q^{|A|-1}(q^{m}-q^{|B|}-q^{|C|}+2q^{|B \cap C|}).
\]
\end{enumerate}
To determine the frequency corresponding to the weight
\[
w_9 = 2(q-1)q^{|A|-1}(q^{m}-q^{|B|}-q^{|C|}+2q^{|B \cap C|}),
\]
we subtract the sum of the frequencies of all other weights from the total number of codewords $q^{2m}$. Consequently, the frequency is 
\[
q^{2m}+q^{2m-|A|-|A\cup(B\cap C)|}-2q^{2m-|A|-|B\cap C|}.
\]
Since the frequency corresponding to Lee weight $0$ is $q^{m-|A|}$, all the above frequencies must be divided by $q^{m-|A|}$ to obtain the Lee weight distribution. \qed
\end{proof}

\begin{theorem}\label{thm:3}
    Let $m\geq 2$ be a positive integer and let $A,B,C \subseteq [m]$ such that       $|C|\leq |B|$. Define $D = \Delta_{A}+ u (\Delta_{B,C})^{c}.$ Then the code $\mathcal{C}_{D}$ defined in \eqref{eqn:C_D code} is an at most nine-weight linear code of length $q^{|A|} (q^{m}-q^{|B|}-q^{|C|}+q^{|B \cap C|})$ and cardinality $q^{m+|A|}$.
    Its Lee weight distribution is given by 
    \begin{table}[H]
\centering
\resizebox{\textwidth}{!}{
\begin{tabular}{l|l}
\hline
Weight & Frequency \\
\hline
$w_{0}=0$ & $f_{0}=1$ \\
\hline
$w_{1}= 2(q-1)q^{m+|A|-1}$ & $f_{1}=q^{m-|A\cup B\cup C|}-1$ \\
\hline
$w_{2}= (q-1)q^{|A|-1}(2q^{m}-q^{|B|}-q^{|C|}+q^{|B\cap C|})$ & $f_{2}=2(q^{m-|B\cup C|}-q^{m-|A\cup B\cup C|})$ \\
\hline
$w_{3}= 2(q-1)q^{|A|-1}(q^{m}-q^{|C|})$ & $f_{3}=q^{m-|A\cup B|}-q^{m-|A\cup B\cup C|}$ \\
\hline
$w_{4}= (q-1)q^{|A|-1}(2q^{m}-q^{|B|}-2q^{|C|}+q^{|B\cap C|})$ & $f_{4}=2(q^{m-|B|}-q^{m-|A\cup B|}-q^{m-|B\cup C|}+q^{m-|A\cup B\cup C|})$ \\
\hline
$w_{5}= 2(q-1)q^{|A|-1}(q^{m}-q^{|B|})$ & $f_{5}=q^{m-|A\cup C|}-q^{m-|A\cup B\cup C|}$ \\
\hline
$w_{6}= (q-1)q^{|A|-1}(2q^{m}-2q^{|B|}-q^{|C|}+q^{|B\cap C|})$ & $f_{6}=2(q^{m-|C|}-q^{m-|A\cup C|}-q^{m-|B\cup C|}+q^{m-|A\cup B\cup C|})$ \\
\hline
$w_{7}= 2(q-1)q^{|A|-1}(q^{m}-q^{|B|}-q^{|C|})$ & $f_{7}=q^{m-|A\cup(B\cap C)|}-q^{m-|A\cup B|}-q^{m-|A\cup C|}+q^{m-|A\cup B\cup C|}$ \\
\hline
$w_{8}= (q-1)q^{|A|-1}(2q^{m}-2q^{|B|}-2q^{|C|}+q^{|B\cap C|})$ &
$
\begin{aligned}
   f_{8}= 2(& q^{m-|B\cap C|}-q^{m-|A\cup(B\cap C)|}-q^{m-|B|}-q^{m-|C|}+q^{m-|A\cup B|} \\[-1ex]
    & +q^{m-|A\cup C|}+q^{m-|B\cup C|}-q^{m-|A\cup B\cup C|})
\end{aligned}$\\
\hline
$w_{9}= 2(q-1)q^{|A|-1}(q^{m}-q^{|B|}-q^{|C|}+q^{|B\cap C|})$ & $f_{9}=q^{m+|A|}-2q^{m-|B\cap C|}+q^{m-|A\cup (B\cap C)|}$ \\
\hline
\end{tabular}
}
\label{tab:4}
\end{table}
\end{theorem}
\begin{proof}
    The proof follows verbatim from Theorem~\ref{thm:2}. \qed
\end{proof}
\section{Structure of the Gray Image \texorpdfstring{$\Phi(\mathcal{C}_D)$}{CD}}\label{sec4}
\subsection{Distance-Optimality}
In this subsection, we investigate the distance-optimality of the Gray images $\Phi(\mathcal{C}_{D})$ of the codes constructed in Section \ref{sec3}.
\begin{theorem}\label{thm:optimal1}
    Let $D$ be as in Theorem~\ref{thm:1} with $A\subseteq  B$ and $f_5 > 0$. Then $\Phi(\mathcal{C}_{D})$ is a $[2(q^{|A|}-1)(q^{|B|}-1), |A|+|B|,2(q-1)(q^{|A|+|B|-1}-q^{|B|-1}-q^{|A|-1})]$-linear code over $\mathbb{F}_{q}$. Moreover, it is distance optimal if    
     \[
|A|\geq
\begin{cases}
2, & \text{if } q=2,\\[2mm]
4, & \text{if } q\geq 5 \text{ and } |A|=|B|,\\[2mm]
3, & \text{otherwise}.
\end{cases}
\]
\begin{proof}
    The length and dimension of $\Phi(\mathcal{C}_{D})$ are clear. The minimum distance is
\[
d=2(q-1)\left(q^{|A|+|B|-1}-q^{|B|-1}-q^{|A|-1}\right).
\]
By Lemma~\ref{lem: optimality}, we have
\[
\sum_{i=0}^{k-1}\left\lceil\frac{d+1}{q^i}\right\rceil\ge n+1
\]
if and only if
 \[
|A|\geq
\begin{cases}
2, & \text{if } q=2,\\[2mm]
4, & \text{if } q\geq 5 \text{ and } |A|=|B|,\\[2mm]
3, & \text{otherwise}.
\end{cases}
\] \qed
\end{proof}
\end{theorem}
\begin{example}
    This example illustrates Theorem \ref{thm:optimal1}. Let $q=2, m=4, A=\{1,2\}\text{ and } B=\{1,2,3\}$. Then $\Phi(\mathcal{C}_{D})$ is a four-weight binary linear code with parameters $[42,5,20]$ and Hamming weight enumerator $x^{42}+15x^{22}y^{20}+12 x^{20}y^{22}+x^{18}y^{24}+3x^{14}y^{28}$. Furthermore, the code is distance-optimal according to the database \cite{Grassl:codetables}.
\end{example}
\begin{theorem}\label{thm:optimal2}
    Let $D$ be as in Theorem~\ref{thm:1a}. Then $\Phi(\mathcal{C}_{D})$ is a $[2(q^{2m}-q^{|A|+|B|}), 2m, 2(q-1)(q^{2m-1}-q^{|A|+|B|-1})]$-linear code over $\mathbb{F}_{q}$. Moreover, it is a near-Griesmer code and distance-optimal if $|A|+|B|>1$.
\end{theorem}
\begin{proof}
      The length and dimension of $\Phi(\mathcal{C}_{D})$ are clear. The minimum distance is
\[
d=2(q-1)\left(q^{2m-1}-q^{|A|+|B|-1}\right).
\]
By the Griesmer bound,
\begin{eqnarray*}
    g_q(2m,d) &=&  \sum\limits_{i=0}^{2m-1} \left\lceil \frac{2(q-1)(q^{2m-1}-q^{|A|+|B|-1})}{q^i} \right\rceil \\
    &=& 2(q^{2m}-q^{|A|+|B|})-1.
\end{eqnarray*}
Therefore, $\Phi(\mathcal{C}_{D})$ is a near-Griesmer code, and by Lemma \ref{lem:near-griesmer}, the conclusion follows. \qed
\end{proof}
\begin{example}
    This example illustrates Theorem \ref{thm:optimal2}. Let $q=2, m=3, A=\{1,2\},$ and $B=\{1,3\}$. Then $\Phi(\mathcal{C}_{D})$ is a three-weight linear code with parameters $[96,6,48]$ and Hamming weight enumerator $x^{96}+58x^{48}y^{48}+4 x^{40}y^{56}+x^{32}y^{64}$. Furthermore, the code is distance-optimal according to the database \cite{Grassl:codetables}.
\end{example}
\begin{theorem}\label{thm:optimal3}
    Let $D$ be as in Theorem~\ref{thm:2} with $f_{7}>0$. Then $\Phi(\mathcal{C}_{D})$ is a $[2q^{|A|}(q^{m}-q^{|B|}-q^{|C|}+2q^{|B\cap C|}), m+|A|,2(q-1)q^{|A|-1}(q^{m}-q^{|B|}-q^{|C|})]$-linear code over $\mathbb{F}_{q}$. Moreover, it is distance optimal if 
   \[
|A|+|C|\geq
\begin{cases}
2^{|A|+|B\cap C|+2}, & \text{if } q=2,\\[2mm]
4q^{|A|+|B\cap C|}+2, & \text{if } q\geq 5 \text{ and } |B|=|C|,\\[2mm]
4q^{|A|+|B\cap C|}+1, & \text{otherwise}.
\end{cases}
\]
\begin{proof}
    The length and dimension of $\Phi(\mathcal{C}_{D})$ are clear. The minimum distance is
\[
d=2(q-1)q^{|A|-1}\left(q^{m}-q^{|B|}-q^{|C|}\right).
\]
By Lemma~\ref{lem: optimality}, we have
\[
\sum_{i=0}^{k-1}\left\lceil\frac{d+1}{q^i}\right\rceil\ge n+1
\]
if and only if
\[
|A|+|C|\ge
\begin{cases}
2^{|A|+|B\cap C|+2}, & \text{if } q=2,\\
4q^{|A|+|B\cap C|}+2, & \text{if } q\geq 5 \text{ and } |B|=|C|,\\
4q^{|A|+|B\cap C|}+1, & \text{otherwise}.
\end{cases}
\] \qed
\end{proof}
\end{theorem}
\begin{example}
    This example illustrates Theorem \ref{thm:optimal3}. Let $q=2, m=8, A=\emptyset, B=\{1,2,3,4\}$, and $C=\{5,6,7,8\}$. Then $\Phi(\mathcal{C}_{D})$ is a two-weight binary linear code with parameters $[452,8,224]$ and Hamming weight enumerator $x^{452}+225x^{228}y^{224}+30 x^{212}y^{240}$. Furthermore, the code is distance-optimal.
\end{example}

\begin{theorem}\label{thm:optimal4}
    Let $D$ be as in Theorem~\ref{thm:3} with $f_{7}>0$. Then $\Phi(\mathcal{C}_{D})$ is a $[2q^{|A|}(q^{m}-q^{|B|}-q^{|C|}+q^{|B\cap C|}), m+|A|,2(q-1)q^{|A|-1}(q^{m}-q^{|B|}-q^{|C|})]$-linear code over $\mathbb{F}_{q}$. Moreover it is distance optimal if 
   \[
|A|+|C|\geq
\begin{cases}
2^{|A|+|B\cap C|+1}, & \text{if } q=2,\\[2mm]
2q^{|A|+|B\cap C|}+2, & \text{if } q\geq 5 \text{ and } |B|=|C|,\\[2mm]
2q^{|A|+|B\cap C|}+1, & \text{otherwise}.
\end{cases}
\]
\begin{proof}
    The length and dimension of $\Phi(\mathcal{C}_{D})$ are clear. The minimum distance is
\[
d=2(q-1)q^{|A|-1}\left(q^{m}-q^{|B|}-q^{|C|}\right).
\]
By Lemma~\ref{lem: optimality}, we have
\[
\sum_{i=0}^{k-1}\left\lceil\frac{d+1}{q^i}\right\rceil\ge n+1
\]
if and only if
\[
|A|+|C|\ge
\begin{cases}
2^{|A|+|B\cap C|+1}, & \text{if } q=2,\\
2q^{|A|+|B\cap C|}+2, & \text{if } q\geq 5 \text{ and } |B|=|C|,\\
2q^{|A|+|B\cap C|}+1, & \text{otherwise}.
\end{cases}
\] \qed
\end{proof}
\end{theorem}
\begin{example}
    This example illustrates Theorem \ref{thm:optimal4}. Let $q=2, m=4, A=\emptyset, B=\{1,2\}$, and $C=\{3,4\}$. Then $\Phi(\mathcal{C}_{D})$ is a two-weight linear code with parameters $[18,4,8]$ and Hamming weight enumerator $x^{18}+9x^{10}y^{8}+6 x^{6}y^{12}$. Furthermore, the code is distance-optimal according to the database \cite{Grassl:codetables}.
\end{example}
\subsection{Minimality}
In this subsection, we give the criterion for the Gray images of codes constructed in Section \ref{sec3} to be minimal.
\begin{theorem}\label{minimality}
Let $\mathcal{C}_{D}$ be as defined in \eqref{eqn:C_D code} and $\Phi(\mathcal{C}_{D})$ be its Gray image. Then
\begin{enumerate}
    \item Let \(D\) be as in Theorem \ref{thm:1}. If \(A\subseteq B\) and \(|A|>1\), then $\Phi(\mathcal{C}_{D})$ is minimal.
    \item  Let $D$ be as in Theorem \ref{thm:1a}. If $|A|+|B|<2m-1$, then $\Phi(\mathcal{C}_{D})$ is minimal.
    \item  Let $D$ be as in Theorem \ref{thm:2}. If $q^{|B|}+q^{|C|} < q^{m-1}$, then $\Phi(\mathcal{C}_{D})$ is minimal.
    \item Let $D$ be as in Theorem \ref{thm:3}. If $q^{|B|}+q^{|C|} < q^{m-1}$, then $\Phi(\mathcal{C}_{D})$ is minimal.

\end{enumerate}
\end{theorem}
\begin{proof}
\begin{enumerate}
    \item By Corollary \ref{cor:1}, $\wt_{\min}=(q-1)(2q^{|A|+|B|-1}-2q^{|A|-1}-2q^{|B|-1})$ and $\wt_{\max}=(q-1)(2q^{|A|+|B|-1}-2q^{|A|-1}).$ Hence, by Lemma \ref{minimal_lemma}, 
    \[\frac{\wt_{\min}}{\wt_{\max}}>\frac{q-1}{q}\iff q^{|A|-1}+q^{|B|}<q^{|A|+|B|-1}\impliedby |A|>1.
    \]
    \item By Theorem \ref{thm:1a}, $\wt_{\min}=(q-1)(2q^{2m-1}-2q^{|A|+|B|-1})$ and $\wt_{\max}=(q-1)(2q^{2m-1}).$ Hence, by Lemma \ref{minimal_lemma}, 
    \[\frac{\wt_{\min}}{\wt_{\max}}>\frac{q-1}{q}\iff q^{|A|+|B|} < q^{2m-1}\iff |A|+|B| < 2m-1.
    \]
    \item  By Theorem \ref{thm:2}, $\wt_{\min}=(q-1)(2q^{m+|A|-1}-2q^{|A|+|B|-1}
    -2q^{|A|+|C|-1})$ and $\wt_{\max}=(q-1)(2q^{m+|A|-1}).$ Hence, by Lemma \ref{minimal_lemma}, 
    \[\frac{\wt_{\min}}{\wt_{\max}}>\frac{q-1}{q}\iff q^{|B|}+q^{|C|} < q^{m-1}.\]
    \item The proof of this part follows along the same lines as that of Part~3.
\end{enumerate} \qed
\end{proof}
    
\begin{example}
This example illustrates Theorem~\ref{minimality} $(1)$. Let $ m = 4,A = B= \{1,2\}$. Then the three-weight code $ \Phi(\mathcal{C}_D)$ is a minimal binary linear code with parameters $ [18, 4,8]$ and Hamming weight enumerator $x^{18} +6x^{10}y^{8} + 6 x^{8}y^{10}+3x^{6}y^{12}$, as verified by Magma.
\end{example}
    
\begin{example}
This example illustrates Theorem~\ref{minimality} $(2)$. Let $ m = 3,A =\{1\}$ and $ B= \{1,2\}$. Then the two-weight code $ \Phi(\mathcal{C}_D)$ is a minimal binary linear code with parameters $ [112, 6, 56]$ and Hamming weight enumerator $x^{112}+56x^{56}y^{56}+7x^{48}y^{64}$, as verified by Magma. 
\end{example}
\begin{example}
This example illustrates Theorem~\ref{minimality} $(3)$. Let $ m = 4,A =\{1\},B=\{1\}$ and $ C= \{2\}$. Then the three-weight code $ \Phi(\mathcal{C}_D)$ is a minimal binary linear code with parameters $[56, 5, 26] $ and Hamming weight enumerator $x^{56}+8x^{30}y^{26}+20x^{28}y^{28}+3x^{24}y^{32}$, as verified by Magma. 
\end{example}
\section{Subfield Codes}\label{sec5}
\begin{definition}\cite{generalized}
   Let $R$ be a subring of the ring $S$ (not necessarily commutative) such that $R$ and $S$ have the same unity. Then a surjective homomorphism of left $R$-modules $\Tr_R^S: S\to R$ is called an \textit{$R$-valued trace} of $S$ if $\ker(\Tr_R^S)$ does not contain any non-zero ideals of $S$.
\end{definition}
Let $R$ be a finite commutative $\mathbb{F}_{q}$-algebra with an $\mathbb{F}_{q}$-valued trace, denoted by $\tau$, and let $\mathcal{B}$ be an $\mathbb{F}_{q}$-basis of $R$. Assume that $\mathcal{C}$ is a linear code of length $n$ over $R$ and $G$ is a generator matrix of $\mathcal{C}$. 
 If we replace each entry of the matrix $G$ by its column representation with respect to the basis $\mathcal{B}$, the resulting matrix generates a code known as the \textit{subfield code}, denoted by $\mathcal{C}^{(q)}$.
 
The following result follows from \cite{Bhagat2024Trace}.
\begin{theorem} \label{trace_matrix}
Suppose $\mathcal{C}$ is a linear code of length $n$ over $R$. Suppose $G=(g_{ij})_{k\times n}$ is a generator matrix of $\mathcal{C}$, $\mathcal{B} = \{  {\alpha}_{1},  {\alpha}_{2},\ldots,  {\alpha}_{m}\} $ is an $\mathbb{F}_{q}$-basis of $R$ and $\tau:R \to \mathbb{F}_q$ is an $\mathbb{F}_q$-valued trace of $R$. Then $\mathcal{C}^{(q)}$ is generated by  
\begin{center}
    $G^{(q)} =  
\left[
\begin{tabular}{c}
     $G_{1}^{(q)}$  \\
      $G_{2}^{(q)}$ \\
      $\vdots$ \\
      $G_{k}^{(q)}$
\end{tabular}
\right],$
\end{center}
where 
\begin{center}
    $G_{j}^{(q)} = \left[ 
    \begin{tabular}{c c c c}
         $\tau(g_{j1}  {\alpha}_{1})$& $\tau(g_{j2}  {\alpha}_{1})$ & $\cdots$ & $\tau(g_{jn}  {\alpha}_{1})$ \\
          $\tau(g_{j1}  {\alpha}_{2})$& $\tau(g_{j2}  {\alpha}_{2})$ & $\cdots$ & $\tau(g_{jn}  {\alpha}_{2})$ \\
          $\vdots$ & $\vdots$ & $\ddots$ & $\vdots$ \\
          $\tau(g_{j1}  {\alpha}_{m})$& $\tau(g_{j2}  {\alpha}_{m})$ & $\cdots$ & $\tau(g_{jn}  {\alpha}_{m})$ \\
    \end{tabular}
    \right]$
\end{center}
for all $j\in [k]$.
\end{theorem}
It is easy to see that the map $\tau : \mathcal{R} \to \mathbb{F}_{q}$ defined by $a+ub \mapsto a+b$ is an $\mathbb{F}_{q}$-valued trace of $\mathcal{R}$ by \cite{Bhagat2024Trace}.
\begin{theorem}
    Let $\mathcal{B}=\{1, u\}$ be an $\mathbb{F}_{q}$-basis  of $\mathcal{R}$. Suppose $\mathcal{C}$ is a linear code over $\mathcal{R}$ and $G =  G_{1}+ uG_{2}$ is its generator matrix, where $G_{i} \in M_{k\times n}(\mathbb{F}_{q})$ for $i=1,2.$ Then $\mathcal{C}^{(q)}$ is a linear code over $\mathbb{F}_{q}$ of length $n$ and is generated by $$G^{(q)} = 
    \left[
    \begin{tabular}{c}
         $G_{1}+G_{2}$ \\
         $G_{1}$
    \end{tabular}
    \right]$$
Furthermore, if $D= D_{1}+ uD_{2}$, then $\mathcal{C}_{D}^{(q)} = \mathcal{C}_{D^{(q)}}$, where $$D^{(q)} = \{ (\bm{d}_{1}+\bm{d}_{2},\bm{d}_{1}) : \bm{d}_{i} \in D_{i}, 1\leq i\leq 2\}.$$
\end{theorem}
\begin{proof}
    Let $ {g}_{jk} =  {g}_{jk}^{(1)} + u {g}_{jk}^{(2)} \in \mathcal{R}$, where $ {g}_{jk}^{(1)}, {g}_{jk}^{(2)} \in \mathbb{F}_{q}$. Then
    \begin{eqnarray*}
        \tau({g}_{jk}) &=& {g}_{jk}^{(1)} + {g}_{jk}^{(2)},  \\
        \tau({g}_{jk} u) &=& {g}_{jk}^{(1)}.
    \end{eqnarray*}
    Using Theorem \ref{trace_matrix}, the result follows. \qed
    \end{proof}
Consider the map
\begin{align*}
    c_{D}^{(q)}: (\mathbb{F}_{q}^{m})^{2} &\longrightarrow \mathcal{C}_{D}^{(q)}\\
    \bm{v} &\longmapsto (\bm{v}\cdot \bm{d})_{\bm{d}\in D^{(q)}},
\end{align*}
where $\bm{v}\cdot \bm{d}= (\bm{a}, \bm{b})\cdot ( \bm{d}_{1}+ \bm{d}_{2}, \bm{d}_{1})$, for $\bm{v}=(\bm{a}, \bm{b})\in(\mathbb{F}_{q}^{m})^2$, $\bm{d}=( \bm{d}_{1}+ \bm{d}_{2}, \bm{d}_{1})\in D^{(q)}$, and $\cdot$ denotes the Euclidean inner product.
Clearly, the map $c_{D}^{(q)}$ is a surjective linear transformation. Thus, 
$$\mathcal{C}_{D}^{(q)} = \{ c_{D}^{(q)}(\bm{a}, \bm{b}) = ((\bm{a}, \bm{b})\cdot ( \bm{d}_{1}+ \bm{d}_{2}, \bm{d}_{1}))_{ \bm{d} \in D^{(q)}} : \bm{a}, \bm{b} \in \mathbb{F}_{q}^{m}\}.$$

Let $n=|D|$. Now
\begin{eqnarray} \label{weight_equation_1}
    \wt_{H}\left(c_{D}^{(q)}(\bm{a}, \bm{b})\right) &=& \wt_{H}\left( ((\bm{a}, \bm{b})\cdot ( \bm{d}_{1}+ \bm{d}_{2}, \bm{d}_{1}))_{ \bm{d}_{1}\in D_{1}, \bm{d}_{2}\in D_{2}}\right) \nonumber\\
    &=& \wt_{H}\left( (\langle\bm{a}+\bm{b},  \bm{d}_{1}\rangle_q + \langle\bm{a},  \bm{d}_{2}\rangle_q )_{ \bm{d}_{1}\in D_{1}, \bm{d}_{2}\in D_{2}} \right) \nonumber \\
    &=& n-\frac{1}{q}\sum\limits_{v\in \mathbb{F}_{p}^{s}}\sum\limits_{\bm{d}_{1}\in D_{1}}\sum\limits_{\bm{d}_{2}\in D_{2}}\zeta_{p}^{\langle v, \langle\bm{a}+\bm{b},\bm{d}_{1}\rangle_{q}+\langle\bm{a},\bm{d}_{2}\rangle_{q}\rangle_{p}} \nonumber \\
    &=&  n-\frac{n}{q}-\frac{1}{q}\sum\limits_{v\in \mathbb{F}_{p}^{s}\setminus\{0\}}\sum\limits_{\bm{d}_{1}\in D_{1}}\zeta_{p}^{\langle v, \langle\bm{a}+\bm{b},\bm{d}_{1}\rangle_{q}\rangle_{p}}\sum\limits_{\bm{d}_{2}\in D_{2}}\zeta_{p}^{\langle v, \langle\bm{a},\bm{d}_{2}\rangle_{q}\rangle_{p}}. \nonumber
\end{eqnarray}
\begin{theorem}\label{thm:subfield1}
    Let $D$ be as in Theorem \ref{thm:1} with $|A|\leq |B|$. Then the subfield code $\mathcal{C}_{D}^{(q)}$ is an at most three-weight linear code over $\mathbb{F}_{q}$ with parameters $[(q^{|A|}-1)(q^{|B|}-1), |A|+|B|, (q-1)(q^{|A|+|B|-1}-q^{|B|-1}-q^{|A|-1})]$. Its Hamming weight distribution is given by
        \begin{table}[H]
\centering
\begin{tabular}{l|l}
\hline
Weight & Frequency \\
\hline
$w_0=0$ & $f_0=1$ \\
\hline
$w_{1}= (q-1)q^{|A|-1}(q^{|B|}-1)$ & $f_1=q^{|A|}-1$ \\
\hline
$w_{2}= (q-1)(q^{|A|}-1)q^{|B|-1}$ & $f_2=q^{|B|}-1$ \\
\hline
$w_{3}= (q-1)(q^{|A|+|B|-1}-q^{|A|-1}-q^{|B|-1})$ & $f_3=q^{|A|+|B|}-q^{|B|}-q^{|A|}+1$ \\
\hline
\end{tabular}

\label{tab:5}
\end{table}
    Moreover, if $|A|=|B|$, then $\mathcal{C}_{D}^{(q)}$ is a near-Griesmer code and, for $|A|>1$, is distance-optimal. If $|A|<|B|$, then $\mathcal{C}_{D}^{(q)}$ is a Griesmer code and hence distance-optimal.
\end{theorem}
\begin{proof}
    Here $D_{1}=\Delta_{A}^*, D_{2} = \Delta_{B}^*$. Consequently, the length of the code $\mathcal{C}_{D}$ is $n = |D| = (q^{|A|}-1)(q^{|B|}-1).$ 
    Now, 
    \[
        \sum\limits_{\bm{d}_{2}\in \Delta_{B}^*}\zeta_{p}^{\langle v, \langle\bm{a},\bm{d}_{2}\rangle_{q}\rangle_{p}}
        = \sum\limits_{\bm{d}_{2}\in \Delta_{B}}\zeta_{p}^{\langle v, \langle\bm{a},\bm{d}_{2}\rangle_{q}\rangle_{p}} -  1,
    \]
    and,
     \[ \sum\limits_{\bm{d}_{1}\in \Delta_{A}^*}\zeta_{p}^{\langle v, \langle\bm{a}+\bm{b},\bm{d}_{1}\rangle_{q}\rangle_{p}}=  \sum\limits_{\bm{d}_{1}\in \Delta_{A}}\zeta_{p}^{\langle v, \langle\bm{a}+\bm{b},\bm{d}_{1}\rangle_{q}\rangle_{p}}-1.
 \]
    Then by Lemma \ref{lem:counting}, for $v\in\mathbb{F}_{p}^{s}\setminus\{0\}$, we have
    \[
    \sum\limits_{\bm{d}_{2}\in \Delta_{B}^*}\zeta_{p}^{\langle v, \langle\bm{a},\bm{d}_{2}\rangle_{q}\rangle_{p}} = 
     \begin{cases}
         q^{|B|}-1, & \text{if } \bm{a} \in \Delta_B^{\perp}; \\
         -1,  & \text{if } \bm{a} \notin \Delta_B^{\perp},
     \end{cases}
    \]
   and,
   \[
   \sum\limits_{\bm{d}_{1}\in \Delta_{A}^*}\zeta_{p}^{\langle v, \langle\bm{a}+\bm{b},\bm{d}_{1}\rangle_{q}\rangle_{p}} = 
     \begin{cases}
         q^{|A|}-1, & \text{if } \bm{a}+\bm{b} \in \Delta_A^{\perp}; \\
         -1,  & \text{if } \bm{a}+\bm{b} \notin \Delta_A^{\perp}.
     \end{cases}
   \]
   There are 2 possible cases which are as follows:
   \begin{enumerate}
       \item[(1)]  $\bm{a}\in \Delta_B^{\perp}$.
       \begin{enumerate}
           \item If $\bm{a}+\bm{b}\in\Delta_{A}^{\perp},$ then
         \[
    \wt_{H}(c_{D}^{(q)}(\bm{v})) = 0.
    \]
    \begin{itemize}
        \item If $\bm{b}\in\Delta_{A}^{\perp}.$ In this case,
         \begin{eqnarray*}
        \# \bm{a} 
        &=& q^{m-|A\cup B|}\\
    \# \bm{b} &=& q^{m-|A|}.
    \end{eqnarray*}
    \item If $\bm{b}\notin\Delta_{A}^{\perp}.$ In this case,
         \begin{eqnarray*}
        \# \bm{a} 
        &=&q^{m-|B|}-q^{m-|A\cup B|}\\
    \# \bm{b} &=& q^{m-|A|}.
    \end{eqnarray*}
    \end{itemize}
    \item If $\bm{a}+\bm{b}\notin\Delta_{A}^{\perp},$ then
         \[
    \wt_{H}(c_{D}^{(q)}(\bm{v})) = (q-1)(q^{|A|+|B|-1}-q^{|A|-1}).
    \]
      \begin{itemize}
        \item $\bm{b}\in\Delta_{A}^{\perp}.$ In this case,
         \begin{eqnarray*}
        \# \bm{a} 
        &=& q^{m-|B|}-q^{m-|A\cup B|}\\
    \# \bm{b} &=& q^{m-|A|}.
    \end{eqnarray*}
    \item $\bm{b}\notin\Delta_{A}^{\perp}.$ 
    \begin{itemize}
        \item $\bm{a}\in\Delta_{A}^{\perp}.$ In this case, \begin{eqnarray*}
        \# \bm{a} 
        &=&q^{m-|A\cup B|}\\
    \# \bm{b} &=& q^{m}-q^{m-|A|}.
    \end{eqnarray*}
     \item $\bm{a}\notin\Delta_{A}^{\perp}.$ In this case, \begin{eqnarray*}
        \# \bm{a} 
        &=&q^{m-|B|}-q^{m-|A\cup B|}\\
    \# \bm{b} &=& q^{m}-2q^{m-|A|}.
    \end{eqnarray*}
    \end{itemize}
        
    \end{itemize}
       \end{enumerate}
        \item[(2)]  $\bm{a}\notin \Delta_B^{\perp}$.
       \begin{enumerate}
           \item If $\bm{a}+\bm{b}\in\Delta_{A}^{\perp},$ then
         \[
    \wt_{H}(c_{D}^{(q)}(\bm{v})) = (q-1)(q^{|A|+|B|-1}-q^{|B|-1}).
    \]
    \begin{itemize}
        \item If $\bm{b}\in\Delta_{A}^{\perp}.$ In this case,
         \begin{eqnarray*}
        \# \bm{a} 
        &=& q^{m-|A|}-q^{m-|A\cup B|}\\
    \# \bm{b} &=& q^{m-|A|}.
    \end{eqnarray*}
    \item If $\bm{b}\notin\Delta_{A}^{\perp}.$ In this case,
         \begin{eqnarray*}
        \# \bm{a} 
        &=&q^{m}-q^{m-|A|}-q^{m-|B|}+q^{m-|A\cup B|}\\
    \# \bm{b} &=& q^{m-|A|}.
    \end{eqnarray*}
    \end{itemize}
    \item If $\bm{a}+\bm{b}\notin\Delta_{A}^{\perp},$ then
         \[
    \wt_{H}(c_{D}^{(q)}(\bm{v})) = (q-1)(q^{|A|+|B|-1}-q^{|B|-1}-q^{|A|-1}).
    \]
       \end{enumerate}
   \end{enumerate}
To determine the frequency corresponding to the weight
\[
(q-1)(q^{|A|+|B|-1}-q^{|B|-1}-q^{|A|-1}),
\]
we subtract the sum of the frequencies of all other weights from the total number of codewords $q^{2m}$. Consequently, the frequency is 
\[
q^{2m}-q^{2m-|A|}-q^{2m-|B|}+q^{2m-|A|-|B|}.
\]
Since the frequency corresponding to Lee weight $0$ is $q^{2m-|A|-|B|}$, all the above frequencies must be divided by $q^{2m-|A|-|B|}$ to obtain the Hamming weight distribution.

Here $n=q^{|A|+|B|}-q^{|A|}-q^{|B|}+1$, $k=|A|+|B|$ and $d=(q-1)(q^{|A|+|B|-1}-q^{|B|-1}-q^{|A|-1})$.
By the Griesmer bound,
\begin{eqnarray*}
    g_q(k,d) &=&  \sum\limits_{i=0}^{|A|+|B|-1} \left\lceil \frac{(q-1)(q^{|A|+|B|-1}-q^{|B|-1}-q^{|A|-1})}{q^i} \right\rceil \\
    &=& \begin{cases}
        q^{|A|+|B|}-q^{|A|}-q^{|B|} , & \text{if } |A|=|B|; \\
         q^{|A|+|B|}-q^{|A|}-q^{|B|}+1,  & \text{if } |A|<|B|.
    \end{cases}
\end{eqnarray*}
Thus by Lemma \ref{lem:near-griesmer}, the conclusion follows. \qed
\end{proof}
\begin{example}
    This example illustrates Theorem \ref{thm:subfield1}. Let $q=3, m=4, A=\{1\}\text{ and } B=\{1,2\}$. Then $\mathcal{C}_{D}^{(3)}$ is a three-weight ternary linear code with parameters $[16,3,10]$ and Hamming weight enumerator $x^{16}+16x^{6}y^{10}+8 x^{4}y^{12}+2y^{16}$. Furthermore, the code is distance-optimal according to the database \cite{Grassl:codetables}.
\end{example}
\begin{theorem}\label{thm:subfield2}
     Let $D$ be as in Theorem \ref{thm:1a}. Then the subfield code $\mathcal{C}_{D}^{(q)}$ is a two-weight linear code over $\mathbb{F}_{q}$ with parameters $[q^{2m}-q^{|A|+|B|}, 2m, (q-1)(q^{2m-1}-q^{|A|+|B|-1})]$. Its Hamming weight distribution is given by
        \begin{table}[H]
\centering
\begin{tabular}{l|l}
\hline
Weight & Frequency \\
\hline
$0$ & $1$ \\
\hline
$(q-1)q^{2m-1}$ & $q^{2m-|A|-|B|}-1$ \\
\hline
$(q-1)(q^{2m-1}-q^{|A|+|B|-1})$ & $q^{2m}-q^{2m-|A|-|B|}$\\
\hline
\end{tabular}
\label{tab:8}
\end{table}
Moreover, $\mathcal{C}_{D}^{(q)}$ is a Griesmer code and hence distance-optimal.
\end{theorem}
\begin{example}
    This example illustrates Theorem \ref{thm:subfield2}. Let $q=2, m=3, A=\{1\}\text{ and } B=\{3\}$. Then $\mathcal{C}_{D}^{(2)}$ is a two-weight binary linear code with parameters $[60,6,30]$ and Hamming weight enumerator $x^{60}+48x^{30}y^{30}+15 x^{28}y^{32}$. Furthermore, the code is distance-optimal according to the database \cite{Grassl:codetables}.
\end{example}
\begin{theorem}\label{thm:subfield3}
    Let $D$ be as in Theorem \ref{thm:2}. Then the subfield code $\mathcal{C}_{D}^{(q)}$ is an at most five-weight linear code over $\mathbb{F}_{q}$ with parameters $[q^{|A|} (q^{m}-q^{|B|}-q^{|C|}+2q^{|B \cap C|}), m+|A|, (q-1)(q^{m+|A|-1}-q^{|A|+|B|-1}-q^{|A|+|C|-1})]$. Its Hamming weight distribution is given by
        \begin{table}[H]
\centering
\begin{tabular}{l|l}
\hline
Weight & Frequency \\
\hline
$0$ & $1$ \\
\hline
$(q-1)q^{m+|A|-1}$ & $q^{m-|B\cup C|}-1$ \\
\hline
$(q-1)q^{|A|-1}(q^{m}-q^{|C|})$ & $q^{m-|B|}-q^{m-|B\cup C|}$ \\
\hline
$(q-1)q^{|A|-1}(q^{m}-q^{|B|})$ & $q^{m-|C|}-q^{m-|B\cup C|}$ \\
\hline
$(q-1)q^{|A|-1}(q^{m}-q^{|B|}-q^{|C|})$ & $q^{m-|B\cap C|}-q^{m-|B|}-q^{m-|C|}+q^{m-|B\cup C|}$ \\
\hline
$(q-1)q^{|A|-1}(q^{m}-q^{|B|}-q^{|C|}+2q^{|B\cap C|})$ & $q^{m+|A|}-q^{m-|B\cap C|}$\\
\hline
\end{tabular}

\label{tab:6}
\end{table}
\end{theorem}
\begin{remark}
    In the above theorem, the code $\mathcal{C}_{D}^{(q)}$ is distance-optimal if $2q^{|A|+|B\cap C|}<|A|+|C|+\lfloor\frac{|C|}{|B|}\rfloor\left(\lfloor\frac{2}{q}\rfloor-1\right)$. This follows directly from \cite[Theorem 10]{Hu2026several}.
\end{remark}
\begin{theorem}
    Let $D$ be as in Theorem \ref{thm:3}. Then the subfield code $\mathcal{C}_{D}^{(q)}$ is an at most five-weight linear code over $\mathbb{F}_{q}$ with parameters $[q^{|A|} (q^{m}-q^{|B|}-q^{|C|}+q^{|B \cap C|}), m+|A|, (q-1)(q^{m+|A|-1}-q^{|A|+|B|-1}-q^{|A|+|C|-1})]$. Its Hamming weight distribution is given by
        \begin{table}[H]
\centering
\begin{tabular}{l|l}
\hline
Weight & Frequency \\
\hline
$0$ & $1$ \\
\hline
$(q-1)q^{m+|A|-1}$ & $q^{m-|B\cup C|}-1$ \\
\hline
$(q-1)q^{|A|-1}(q^{m}-q^{|C|})$ & $q^{m-|B|}-q^{m-|B\cup C|}$ \\
\hline
$(q-1)q^{|A|-1}(q^{m}-q^{|B|})$ & $q^{m-|C|}-q^{m-|B\cup C|}$ \\
\hline
$(q-1)q^{|A|-1}(q^{m}-q^{|B|}-q^{|C|})$ & $q^{m-|B\cap C|}-q^{m-|B|}-q^{m-|C|}+q^{m-|B\cup C|}$ \\
\hline
$(q-1)q^{|A|-1}(q^{m}-q^{|B|}-q^{|C|}+q^{|B\cap C|})$ & $q^{m+|A|}-q^{m-|B\cap C|}$\\
\hline
\end{tabular}

\label{tab:7}
\end{table}
Moreover, the code $\mathcal{C}_{D}^{(q)}$ is distance-optimal if $q^{|A|+|B\cap C|}<|A|+|C|+\lfloor\frac{|C|}{|B|}\rfloor\left(\lfloor\frac{2}{q}\rfloor-1\right)$.
\end{theorem}
\subsection{Minimality of Subfield Codes}
In this subsection, we give the criterion for the subfield codes constructed in this article to be minimal.
\begin{theorem}\label{minimality2}
Let $\mathcal{C}_{D}$ be as defined in \eqref{eqn:C_D code} and $\mathcal{C}_{D}^{(q)}$ be the subfield code. Then
\begin{enumerate}
    \item Let \(D\) be as in Theorem \ref{thm:1} and $|A|\leq |B|$. If $|A|>1,$ then $\mathcal{C}_{D}^{(q)}$ is minimal.
    \item  Let $D$ be as in Theorem \ref{thm:1a}. If $|A|+|B|<2m-1$, then $\mathcal{C}_{D}^{(q)}$ is minimal.
    \item  Let $D$ be as in Theorem \ref{thm:2}. If $q^{|B|}+q^{|C|} < q^{m-1}$, then $\mathcal{C}_{D}^{(q)}$ is minimal.
    \item Let $D$ be as in Theorem \ref{thm:3}. If $q^{|B|}+q^{|C|} < q^{m-1}$, then $\mathcal{C}_{D}^{(q)}$ is minimal.

\end{enumerate}
\end{theorem}
\begin{proof}
\begin{enumerate}
    \item By Theorem \ref{thm:subfield1}, $\wt_{\min}=(q-1)(q^{|A|+|B|-1}-q^{|A|-1}-q^{|B|-1})$ and $\wt_{\max}=(q-1)(q^{|A|+|B|-1}-q^{|A|-1}).$ Hence, by Lemma \ref{minimal_lemma}, 
    \[\frac{\wt_{\min}}{\wt_{\max}}>\frac{q-1}{q}\iff q^{|A|-1}+q^{|B|}<q^{|A|+|B|-1}\impliedby |A|>1.
    \]
    \item By Theorem \ref{thm:subfield2}, $\wt_{\min}=(q-1)(q^{2m-1}-q^{|A|+|B|-1})$ and $\wt_{\max}=(q-1)(q^{2m-1}).$ Hence, by Lemma \ref{minimal_lemma}, 
    \[\frac{\wt_{\min}}{\wt_{\max}}>\frac{q-1}{q}\iff q^{|A|+|B|} < q^{2m-1}\iff |A|+|B| < 2m-1.
    \]
    \item  By Theorem \ref{thm:subfield3}, $\wt_{\min}=(q-1)(q^{m+|A|-1}-q^{|A|+|B|-1}
    -q^{|A|+|C|-1})$ and $\wt_{\max}=(q-1)(q^{m+|A|-1}).$ Hence, by Lemma \ref{minimal_lemma}, 
    \[\frac{\wt_{\min}}{\wt_{\max}}>\frac{q-1}{q}\iff q^{|B|}+q^{|C|} < q^{m-1}.\]
    \item The proof of this part follows along the same lines as that of Part~3.
\end{enumerate} \qed
\end{proof}
\begin{example}
This example illustrates Theorem~\ref{minimality2} $(1)$. Let $ q=3, m = 4,A =\{1,2\}, B= \{2,4\}$. Then the two-weight code $ \mathcal{C}_D^{(q)}$ is a minimal ternary linear code with parameters $ [64, 4, 42]$ and Hamming weight enumerator $x^{64} +64x^{22}y^{42} + 16 x^{16}y^{48}$, as verified by Magma.
\end{example}

\section{Conclusion}\label{sec6}
In this work, we have investigated four classes of linear codes over the ring $\mathbb{F}_q+u\mathbb{F}_q$, where $u^2=0$, constructed using simplicial complexes with either one or two maximal elements, and explicitly determined their Lee weight distributions. Furthermore, with the help of a Gray map, we have obtained several families of optimal few-weight codes over $\mathbb{F}_q$ and established sufficient criteria ensuring their minimality. In addition, we have studied the subfield codes arising from our construction and derived mild sufficient conditions under which they are distance-optimal and minimal. Several examples have been presented to illustrate the theoretical results, and their parameters have been verified using MAGMA.

A natural avenue for future research is to seek other defining sets over $\mathbb{F}_{q}+u\mathbb{F}_{q}$, where $u^2=0$, or over other finite alphabets that lead to infinite families of distance-optimal codes.

\section*{Acknowledgements}
The first author expresses gratitude to MHRD, India, for financial support in the form of a Senior Research Fellowship at the Indian Institute of Technology Delhi. The second author acknowledges the Prime Minister's Research Fellowship
(PMRF ID: 1403187) for financial support. The authors also express their gratitude to Dr Anuj Kumar Bhagat for helpful discussions.
\section*{Declarations}
\subsection*{Conflict of Interest}
All authors declare that they have no conflict of interest.

\bibliographystyle{abbrv}
\bibliography{ref}
\end{document}